\documentclass[11pt,a4paper
]{article}
\usepackage{amsthm,amsfonts,amsmath,amscd,amssymb}
\usepackage{latexsym}
\usepackage{euscript}
\usepackage{enumitem}
\usepackage{graphicx}
\usepackage{tikz}
\usepackage{caption}
\usepackage{subcaption}
\usepackage{cite}
\usepackage[utf8]{inputenc}
\usepackage[english]{babel}
\usepackage{xcolor}

\usepackage{hyperref}
\usepackage{etoolbox}
\usepackage{forloop}

\hypersetup{colorlinks=true}
\numberwithin{equation}{section}
\numberwithin{figure}{section}
\numberwithin{table}{section}

\newcommand{\key}[1]{\par\vskip 0.5em\small{\slshape Key words:\/} #1}

\newcommand{\msc}[1]{\par\vskip 0.5em\small{\slshape Mathematical Subject Classification 2010:\/} #1}

\newcommand\arXiv[1]{\url{https://arxiv.org/abs/#1}}

\theoremstyle{plain}
\newtheorem{theorem}{Theorem}
\newtheorem{corollary}[theorem]{Corollary}
\newtheorem{proposition}[theorem]{Proposition}
\newtheorem{lemma}[theorem]{Lemma}
\theoremstyle{definition}

\theoremstyle{remark}

\newtheorem{remark}[theorem]{Remark}

\usepackage{bm,bbm}

\title
{On the Correlation Functions of the Characteristic Polynomials of Random Matrices with Independent Entries: Interpolation Between Complex and Real Cases}

\author{Ievgenii Afanasiev\footnote{B. Verkin Institute for Low Temperature Physics and Engineering of the National Academy of Sciences of Ukraine, 47 Nauky Ave., Kharkiv, 61103, Ukraine. E-mail: \href{mailto:afanasiev@ilt.kharkov.ua}{afanasiev@ilt.kharkov.ua}, \href{mailto:afanasiev@ilt.kharkov.ua}{ie.afanasiev@gmail.com}}}

\begin{document}
\maketitle



\newcommand{\Compl}{\mathbb{C}}
\newcommand{\R}{\mathbb{R}}
\newcommand{\Z}{\mathbb{Z}}

\newcommand{\conseq}{\Rightarrow}

\newcommand{\Mat}{\mathrm{Mat}}
\newcommand{\conj}[1]{\overline{#1}}
\newcommand{\transp}[1]{#1^T}
\newcommand{\rev}[1]{\dual{#1}}
\newcommand{\dual}[1]{#1^R}
\newcommand{\cConjScl}[1]{\bar{#1}}
\newcommand{\cConjMat}[1]{#1^*}
\newcommand{\aConjScl}[1]{#1^*}
\newcommand{\aConjMat}[1]{#1^+}
\newcommand{\ds}[1]{\check{#1}}

\newcommand{\setcharf}[1]{\mathbbm{1}_{#1}}
\newcommand{\charfp}{\psi}
\newcommand{\asKer}[2]{K\left(#1, #2\right)}
\newcommand{\CF}{\mathsf{f}}

\newcommand{\Gin}{\mathrm{Gin}}
\newcommand{\ens}{M_n}
\newcommand{\sclA}{y}
\newcommand{\matA}{Y}

\newcommand{\cSclGin}{x}
\newcommand{\cMatGin}{X}
\newcommand{\cMatPos}{\mathcal{Z}}
\newcommand{\cSclN}{t}
\newcommand{\cVecN}{\bm{\cSclN}}
\newcommand{\cSclA}{q}
\newcommand{\cMatA}{Q}
\newcommand{\cSetMatA}{\bm{Q}}
\newcommand{\cVecGI}{\bm{h}}
\newcommand{\cSclB}{a}

\newcommand{\aSclA}{\xi}
\newcommand{\aVecA}{\bm{\aSclA}}
\newcommand{\aMatA}{\Xi}
\newcommand{\aSetMatA}{\mathbf{\aMatA}}
\newcommand{\aSclB}{\phi}
\newcommand{\aVecB}{\bm{\aSclB}}
\newcommand{\aMatB}{\Phi}
\newcommand{\aSclBt}{\varphi}
\newcommand{\aVecBt}{\aVecB}
\newcommand{\aSclC}{\theta}
\newcommand{\aVecC}{\bm{\aSclC}}
\newcommand{\aMatC}{\Theta}
\newcommand{\aSclCt}{\vartheta}
\newcommand{\aVecCt}{\aVecC}
\newcommand{\aSclD}{\rho}
\newcommand{\aVecD}{\bm{\aSclD}}
\newcommand{\aSclE}{\tau}
\newcommand{\aSclF}{\nu}
\newcommand{\aVecF}{\bm{\aSclF}}
\newcommand{\aSclG}{\upsilon}
\newcommand{\aVecG}{\bm{\aSclG}}
\newcommand{\aMatG}{\Upsilon}
\newcommand{\aSclH}{\aSclE}
\newcommand{\aVecH}{\bm{\aSclH}}
\newcommand{\aVecGrI}{\bm{\upsilon}}

\newcommand{\tempcumul}{\kappa}
\newcommand{\cumul}[2]{\tempcumul_{#1,#2}}
\newcommand{\realcumul}{\tempcumul}
\newcommand{\seccumul}{\cumul{2}{0}}
\newcommand{\indexset}{\mathcal{I}}
\newcommand{\emptyindex}{\varnothing}
\newcommand{\diLambda}{\mathcal{M}}
\newcommand{\stpointsnbh}{\Omega_n}
\newcommand{\idom}{\mathcal{D}}
\newcommand{\Vanddet}{\triangle}
\newcommand{\herm}{\mathcal{H}}
\newcommand{\USp}{\mathrm{USp}}
\newcommand{\scoeff}{\mathsf{c}}
\newcommand{\vol}{\mathrm{vol}}
\newcommand{\skB}{\mathcal{B}}
\newcommand{\Hess}{\mathsf{H}}
\newcommand{\Fperm}{\mathcal{F}}

\newcommand{\partder}[2]{\frac{\partial #1}{\partial #2}}
\newcommand{\der}[2]{\frac{d #1}{d #2}}

\newcommand{\abs}[1]{\left\lvert#1\right\rvert}
\newcommand{\abssized}[2][ ]{#1\lvert#2#1\rvert}
\newcommand{\norm}[1]{\left\lVert#1\right\rVert}
\newcommand{\normsized}[2][ ]{#1\lVert#2#1\rVert}

\newcommand{\tr}{\mathop{\mathrm{tr}}}
\newcommand{\Pf}{\mathop{\mathrm{Pf}}}
\newcommand{\E}{\operatorname{\mathbf{E}}}
\newcommand{\diag}{\mathop{\mathrm{diag}}}



\begin{abstract}
The paper is concerned with the correlation functions of the characteristic polynomials of random matrices with independent complex entries. We investigate how the asymptotic behavior of the correlation functions depends on the second moment of the common probability law of the matrix entries, a sort of ``reality measure'' of the entries. It is shown that the correlation functions behave like that for the Complex Ginibre Ensemble up to a~factor depending only on the second moment and the fourth absolute moment of the common probability law of the matrix entries.

\key{Random matrix theory, Ginibre ensemble, correlation functions of characteristic polynomials, moments of characteristic polynomials, SUSY.}

\msc{60B20, 15B52.}
\end{abstract}


\section{Introduction}
The ensemble of random matrices with independent entries was introduced by Ginibre in 1965~\cite{Gin:65}. To be exact, he introduced a partial case when entries of the matrices have Gaussian distribution. Anyway, the ensemble appeared to be significant and has been attracting scientists' attention since that time.

Random matrices with independent entries are usually considered over three fields: complex numbers, real numbers and quaternions. An asymptotic behavior of the correlation functions of the characteristic polynomials was recently computed in the complex case~\cite{Af:19} and in the real case~\cite{Af:20}. The goal of the current article is to obtain a similar result in the intermediate case between the complex and the real ones.

Let us proceed to precise definitions. We consider the matrices of the form
\begin{equation} \label{Gin ens}
\ens = \frac{1}{\sqrt{n}}\cMatGin = \frac{1}{\sqrt{n}}(x_{jk})_{j,k = 1}^n,
\end{equation}
where $x_{jk}$ are i.i.d.\ complex random variables such that
\begin{equation}\label{moments}
\E\{x_{jk}\} = 0, \quad \E\{\abs{x_{jk}}^2\} = 1, \quad \E\lbrace x_{jk}^2 \rbrace =: \seccumul.
\end{equation}
Here and everywhere below $\E{}$ denotes an expectation with respect to (w.r.t.)\ all random variables. In the particular case if the entries $x_{jk}$ are complex or real Gaussian this ensemble is known as Complex or Real Ginibre Ensemble respectively ($\Gin(\Compl)$ resp.\ $\Gin(\R)$). The parameter $\seccumul$ plays a role of a ``reality measure''. Indeed, on the one hand $\seccumul = 0$ in the complex case. Om the other hand $\seccumul = 1$ in the real case.

Notice that the ensemble~\eqref{Gin ens} has various applications in physics, neuroscience, economics, etc. For detailed information see~\cite{Ak-Ka:07} and references therein.


Define the Normalized Counting Measure (NCM) of eigenvalues as
\begin{equation*}
	N_n(\Delta) = \#\{\lambda_j^{(n)} \in \Delta,\, j = 1, \ldots, n\}/n, 
\end{equation*}
where $\Delta$ is an arbitrary Borel set in the complex plane, $\left\{\lambda_j^{(n)}\right\}_{j = 1}^n$ are the eigenvalues of~$\ens$. The NCM is known to converge to the uniform distribution on the unit disc. This distribution is called the circular law. This result has a long and rich history. Mehta was the first who obtained it for $x_{jk}$ being complex Gaussian in 1967~\cite{Me:67}. The proof strongly relied on the explicit formula for the common probability density of eigenvalues due to Ginibre~\cite{Gin:65}. Unfortunately, there is no such a formula in the general case. That is why other methods have to be used. The Hermitization approach introduced by Girko~\cite{Gir:84} appeared to be an effective method. The main idea is to reduce the study of matrices~\eqref{Gin ens} to the study of Hermitian matrices using the logarithmic potential of a measure
\begin{equation*}
	P_\mu(z) = \int\limits_\Compl \log \abs{z - \zeta}\, d\mu(\zeta).
\end{equation*}
This approach was successfully developed by Girko in the next series of works~\cite{Gir:94, Gir:04_1, Gir:04_2, Gir:05}. The final result in the most general case was established by Tao and Vu~\cite{Ta-Vu:10}. Notice that there are a lot of partial results besides those listed above. The interested reader is directed to~\cite{Bo-Ch:12}.

The Central Limit Theorem (CLT) for linear statistics of non-Hermitian random matrices of the form~\eqref{Gin ens} was first proven for radial-invariant test functions in the complex case by Forrester~\cite{Fo:99}. The study was continued in the complex case by Rider and Silverstein~\cite{Ri-Si:06}, Rider and Virag~\cite{Ri-Vi:07}, in the real case by O'Rourke and Renfrew~\cite{OR-Re:16}, in both cases by Tao and Vu~\cite{Ta-Vu:15} and Kopel \cite{Ko:15}. The best result for today was obtained by Cipolloni, Erd{\H o}s and Schr\"oder for the complex case in~\cite{Ci-Er-Sc:21_cCLT} and for the real case in~\cite{Ci-Er-Sc:21_rCLT}. They proved CLT for a bit more than twice differentiable test functions assuming that the common distribution of matrix entries has finite moments. 

A local regime for matrices \eqref{Gin ens} is worse studied. The asymptotic behavior of the $k$-point correlation function for Ginibre ensembles is well-known, see \cite{Gin:65, Me:67} for $\Gin(\Compl)$ and \cite{Ed:97,Fo-Na:07,Bo-Si:09} for $\Gin(\R)$. A general distribution case was considered in~\cite{Ta-Vu:15}. It was established in both cases that the $k$-point correlation function converges in vague topology to that for Ginibre ensemble if $x_{jk}$ having the first four moments as in the Gaussian case. The condition of matching moments was recently overcome at the edge of the spectrum (i.e.\ at $\abs{z} = 1$) in~\cite{Ci-Er-Sc:22}. The last result strongly relies on an estimate for the least singular value obtained in~\cite{Ci-Er-Sc:20} using the  supersymmetry technique (SUSY).

One can observe that non-Hermitian random matrices are more complicated than their Hermitian counterparts. Indeed, the Hermitian case was successfully dealt with using the Stieltjes transform or the moments method. However, a measure in the plane can not be recovered from its Stieltjes transform or its moments. Thus these approaches to the analysis fail in the non-Hermitian case. 

The present article suggests to use the 
SUSY
. It is a rather powerful method which is widely applied at the physical level of rigor (for instance \cite{Br-Hi:00, Br-Hi:01, Fy:02, Fy-St:03, Zi:06, Fy-Mi:91, Mi-Fy:91}). There are also a lot of rigorous results, which were obtained using SUSY in the recent years, e.g.\ \cite{Di-Lo-So:21, Di-Sp-Zi:10, ShM-ShT:16, ShM-ShT:17, ShM-ShT:18, Ci-Er-Sc:20, Ba-Er:17, Di-La:17, Di-La:20, Di-Me-Ro:14, Sha:13} etc. Supersymmetry technique is usually used in order to obtain an integral representation for ratios of determinants. Since the main spectral characteristics such as density of states, spectral correlation functions, etc.\ often can be expressed via ratios of determinants, SUSY allows to get the integral representation for these characteristics too. For detailed discussion on connection between spectral characteristics and ratios of determinants see \cite{St-Fy:03, Bo-St:06, Gu:15}. See also \cite{Fy-St:03, Re-Ki-Gu-Zi:12}.

Let us consider the second spectral correlation function $R_2$ defined by the equality
\begin{equation*}
	\E\Bigg\{2\sum\limits_{1 \le j_1 < j_2 \le n}\eta\left(\lambda_{j_1}^{(n)}, \lambda_{j_2}^{(n)}\right)\Bigg\} = \int\limits_{\Compl^2} \eta(\lambda_1, \lambda_2) R_2(\lambda_1, \lambda_2) d\bar{\lambda}_1 d\lambda_1 d\bar{\lambda}_2 d\lambda_2,
\end{equation*}
where the function $\eta \colon \Compl^2 \to \Compl$ is bounded, continuous and symmetric in its arguments. Using the logarithmic potential, $R_2$ can be represented via ratios of the determinants of $\ens$ with the most singular term of the form
\begin{equation}\label{2detsratio}
\int\limits_{0}^{\varepsilon_0}\int\limits_{0}^{\varepsilon_0}\frac{\partial^2}{\partial \delta_1 \partial \delta_2}\E\Bigg\{
\prod\limits_{j = 1}^2 \frac{\det \left((\ens - z_j)(\ens - z_j)^* + \delta_j\right)}{\det \left((\ens - z_j)(\ens - z_j)^* + \varepsilon_j\right)}
\Bigg\}\Bigg|_{\delta = \varepsilon} d\varepsilon_1 d\varepsilon_2
\end{equation}
The integral representation for \eqref{2detsratio} obtained by SUSY will contain both commuting and anti-commuting variables. Such type integrals are rather difficult to analyse. That is why one would investigate a more simple but similar integral to shed light on the situation. This integral arises from the study of the correlation functions of the characteristic polynomials. Moreover, the correlation functions of the characteristic polynomials are of independent interest. They were studied for many ensembles of Hermitian and real symmetric matrices, for instance \cite{Br-Hi:00, Br-Hi:01, ShT:11, ShT:13, ShM-ShT:17, Af:16} etc. The other result on the asymptotic behavior of the correlation functions of the characteristic polynomials of non-Hermitian matrices of the form $H + i\Gamma$, where $H$ is from Gaussian Unitary Ensemble (GUE) and $\Gamma$ is a fixed matrix of rank $M$, was obtained in \cite{Fy-Kh:99}. The kernel computed there, in the limit of rank $M \to \infty$ of the perturbation $\Gamma$ (taken after matrix size $n \to \infty$) after appropriate rescaling approaches the form \eqref{asympt ker}. It was demonstrated in \cite[Sec. 2.2]{Fy-So:03}.

Let us introduce the $m$\textsuperscript{th} correlation function of the characteristic polynomials
\begin{equation}
\label{F_m}
\CF_m(Z) = \E\Bigg\{
\prod\limits_{j = 1}^m \det \left(\ens - z_j\right)\left(\ens - z_j\right)^*
\Bigg\},
\end{equation}
where 
\begin{equation}\label{Z def}
Z = \diag\{z_1, \dotsc, z_m\}
\end{equation}
and $z_1$, \ldots, $z_m$ are complex parameters which may depend on~$n$. We are interested in the asymptotic behavior of \eqref{F_m}, as $n \to \infty$, for
\begin{equation}\label{z_j}
z_j = z_0 + \frac{\zeta_j}{\sqrt{n}}, \quad j = 1,2,\dotsc,m,
\end{equation}
where $z_0$ is either in the bulk ($\abs{z_0} < 1$) or at the edge ($\abs{z_0} = 1$) of the spectrum and $\zeta_1$, \ldots, $\zeta_m$ are $n$-independent complex numbers. 
The functions~\eqref{F_m} are well-studied for the Complex Ginibre Ensemble, see~\cite{Ak-Ve:03, We-Wo:19}. A general distribution case was considered in \cite{Af:19, Af:20}. It was showed that in the complex case for any $z_0$ in the unit disk
\begin{equation}\label{complex case}
	\lim\limits_{n \to \infty} n^{-\frac{m^2 - m}{2}}\frac{\CF_m(Z)}{\CF_1(z_1)\dotsb \CF_1(z_m)} = 
	e^{\frac{m^2 - m}{2}\left(1 - \abs{z_0}^2\right)^2\cumul{2}{2}} \frac{\det (K_\Compl(\zeta_j, \zeta_k))_{j,k = 1}^m}{\abs{\Vanddet(\cMatPos)}^2},
\end{equation}
where $\cumul{2}{2} = \E\{\abs{x_{11}}^4\} - \abs{\E\lbrace x_{11}^2 \rbrace}^2 - 2$ and 
\begin{align}
\label{asympt ker}
K_\Compl(z, w) &= e^{-\abs{z}^2/2 - \abs{w}^2/2 + z\conj{w}}, \\
\label{cMatPos def}
\cMatPos &= \diag\{\zeta_1, \dotsc, \zeta_m\},
\end{align}
and $\Vanddet(\cMatPos)$ is a Vandermonde determinant of $\zeta_1$, \ldots, $\zeta_m$. Whereas in the real case for any 
$z_0 \in [-1, 1]$
\begin{equation*} 
	\lim\limits_{n \to \infty} n^{-2}\frac{\CF_2(Z)}{\CF_1(z_1)\CF_1(z_2)} = C
	e^{\left(1 - \abs{z_0}^2\right)^2\cumul{2}{2}} \frac{\Pf (K_\R(\zeta_j, \zeta_k))_{j,k = 1}^2}{\Vanddet(\zeta_1, \zeta_2, \cConjScl{\zeta}_1, \cConjScl{\zeta}_2)\vphantom{\tilde{\zeta}}},
\end{equation*}
where
\begin{equation*}
	K_{\R}(\zeta_j, \zeta_k) = e^{-\frac{\abssized[]{\zeta_j}^2}{2} - \frac{\abssized{\vphantom{\zeta_j}\zeta_k}^2}{2}}\begin{pmatrix}
		(\zeta_j - \zeta_k)e^{\zeta_j\zeta_k} & (\zeta_j - \cConjScl{\zeta}_k)e^{\zeta_j\cConjScl{\zeta}_k} \\
		(\cConjScl{\zeta}_j - \zeta_k)e^{\cConjScl{\zeta}_j\zeta_k} & (\cConjScl{\zeta}_j - \cConjScl{\zeta}_k)e^{\cConjScl{\zeta}_j\cConjScl{\zeta}_k}
	\end{pmatrix}.
\end{equation*}
%
%
%
%
In the current paper we extend the results of~\cite{Af:19, Af:20} to the case of arbitrary $\seccumul$, $\abs{\seccumul} \le 1$. The main result is
\begin{theorem}\label{th1}
	Let an ensemble of real random matrices $\ens$ be defined by~\eqref{Gin ens} and~\eqref{moments}. Let the first $2m$ moments ($m > 1$) of the common distribution of entries of $\ens$ be finite and $z_j$, $j = 1, \dotsc, m$, have the form \eqref{z_j}. Let also $z_0$ and $\seccumul$ satisfy at least one of two following conditions
	\begin{enumerate}[label={\upshape(\roman*)},ref={\upshape(\roman*)}]
	\item\label{th1(i)} $\abs{\seccumul} < 1$ and $\abs{z_0} < 1$;
	\item\label{th1(ii)} 
	$\abs{\seccumul} = 1$ and $\abs{z_0} < 1$, $z_0 \notin \R$.
	\end{enumerate}
	Then the $m$\textsuperscript{th} correlation function of the characteristic polynomials~\eqref{F_m} satisfies the asymptotic relation
	\begin{equation}\label{main result1}
		\begin{split}
		\lim\limits_{n \to \infty} n^{-\frac{m^2 - m}{2}} \frac{\CF_m(Z)}{\CF_1(z_1)\dotsb \CF_1(z_m)} 
		= C_{m,z_0}
		e^{d(\seccumul, \cumul{2}{2})} \frac{\det (K_\Compl(\zeta_j, \zeta_k))_{j,k = 1}^m}{\abs{\Vanddet(\cMatPos)}^2},
		\end{split}
	\end{equation}
	where $C_{m,z_0}$ is some constant, which does not depend on the common distribution of entries and on $\zeta_1$, \ldots, $\zeta_m$; 
	$\cumul{2}{2} = \E\{\abs{x_{11}}^4\} - \abs{\E\lbrace x_{11}^2 \rbrace}^2 - 2$,
	\begin{align}
	\label{d def}
	\begin{split}
	d(\seccumul, \cumul{2}{2}) &= -m\log\Bigl\lbrace \abs{1 - \abs{\seccumul}z_0^2}^2 - \abs{\seccumul}^2\bigl(1 - \abs{z_0}^2\bigr)^2\Bigr\rbrace \\
	&\quad{}+ \frac{m^2 - m}{2}\bigl(1 - \abs{z_0}^2\bigr)^2\cumul{2}{2},
	\end{split} 
	\end{align}
	$\Vanddet(\cMatPos)$ is a Vandermonde determinant of $\zeta_1$, \ldots, $\zeta_m$ 
	and $K_\Compl(z, w)$ is defined in~\eqref{asympt ker}.
\end{theorem}

Note that \eqref{main result1} has an additional factor compared with \eqref{complex case}. This factor shows the dependence of the asymptotics of $\CF_m$ (here and below we omit $Z$ only if $Z = \diag\{z_1, \dotsc, z_m\}$) on $\seccumul$.

The paper is organized as follows. Section~\ref{sec:IR} discusses a suitable integral representation for $\CF_m$. In Section~\ref{sec:asympt analysis} we apply the steepest descent method to the suitable integral representation and find out the asymptotic behavior of $\CF_m$. In order to compute it, the Harish-Chandra/Itsykson--Zuber formula is used. 
For the reader convenience the latter section is divided into two parts. The first part deals with a simpler partial case whereas the second one treats a general case.

\subsection{Notations}
Through out the article lower-case letters denote scalars, bold lower-case letters denote vectors, upper-case letters denote matrices and bold upper-case letters denote sets of matrices. We use the same letter for a matrix, for its columns and for its entries. 
Table \ref{tab:notation} shows the exact correspondence.
\begin{table}
	\begin{center}
		\begin{tabular}{|c|c|c|c|}
			\hline
			Set of matrices & Matrix & Vector & Entry \\
			\hline
			$\cSetMatA$ & $\cMatA_{p,s}$ & & $\cSclA^{(p,s)}_{\alpha\beta}$ \\
			\hline
			& 
			& $\aVecB$ & $\aSclB_{j}$ \\
			\hline
			& 
			& $\aVecC$ & $\aSclC_{j}$ \\
			\hline
			& $\matA_{p,s}$ & & $\sclA^{(p,s)}_{\alpha\beta}$ \\
			\hline
			& $U$ & & $u_{kj}$ \\
			\hline
			& $V$ & & $v_{kj}$ \\
			\hline
		\end{tabular}
	\end{center}
	\caption{Notation correspondence}\label{tab:notation}
\end{table}
Besides, for any matrix $A$ we denote by $(A)_j$ its $j$-th column and by $(A)_{kj}$ its entry in the $k$-th row and in the $j$-th column.

The term ``Grassmann variable'' is a synonym for ``anti-commuting variable''. The variables of integration $\aSclB$, 
$\aSclC$ 
and $\aSclD$ 
are Grassmann variables, all the other variables of integration unspecified by an integration domain are either complex or real. We split all the generators of Grassmann algebra into two equal sets and consider the generators from the second set as ``conjugates'' of that from the first set. I.e., for Grassmann variable $\aSclG$ we use $\aConjScl{\aSclG}$ to denote its ``conjugate''. Furthermore, if $\aMatG = (\aSclG_{jk})$ means a matrix of Grassmann variables then $\aConjMat{\aMatG}$ is a matrix $(\aConjScl{\aSclG_{kj}})$. $d$-dimensional vectors are identified with $d \times 1$~matrices.

Integrals without limits denote either integration over Grassmann variables or integration over the whole space $\Compl^d$ or $\R^d$. Let also $d\cVecN^*d\cVecN$ ($\cVecN = (\cSclN_1, \dotsc, \cSclN_d)^T \in \Compl^d$) denote the measure $\prod\limits_{j = 1}^d d\cConjScl{\cSclN}_j d\cSclN_j$ on the space $\Compl^d$. Similarly, for vectors with anti-commuting entries $d\aConjMat{\aVecH}d\aVecH = \prod\limits_{j = 1}^d d\aConjScl{\aSclH_j} d\aSclH_j$. Note that the space of matrices is a linear space over $\Compl$. Thus the same notations are used for matrices as well.

$\langle \cdot, \cdot \rangle$ denotes a standard scalar product on $\Compl^d$. For matrices $\langle A, B \rangle = \tr \cConjMat{B}A$. For sets of matrices $\langle \bm{A}, \bm{B} \rangle = \sum_j \langle A_j, B_j \rangle$. The norm we use is defined by $\norm{\cdot} = \sqrt{\langle \cdot, \cdot \rangle}$.

$\binom{m}{p} \times \binom{m}{s}$ matrices appear in the statement of Proposition \ref{prop:IR}. It is natural to number rows and columns of such matrices by subsets of a $m$-element set. To this end, set
\begin{equation} \label{indexset def}
\indexset_{m, p'} = \{\alpha \in \Z^{p'} \mid 1 \leq \alpha_1 < \ldots < \alpha_{p'} \leq m\}.
\end{equation}
If $p' = 0$ we define $\indexset_{m, p'}$ as $\{\emptyindex\}$.

The cumulants $\cumul{p}{s}$ are defined as follows. Consider the function
\begin{equation*}
	\charfp(t_1, t_2) := \E \left\{
	e^{t_1x_{11} + t_2\conj{x}_{11}}
	\right\}.
\end{equation*}
Then
\begin{align}\label{cumul def}
	\cumul{p}{s} = \left.\frac{\partial^{p + s}}{\partial^p t_1 \partial^s t_2} \log\charfp(t_1, t_2)\right|_{t_1 = t_2 = 0}.
\end{align}
In particular, $\cumul{2}{2} = \E\{\abs{x_{11}}^4\} - \abs{\E\lbrace x_{11}^2 \rbrace}^2 - 2$.

Through out the article $U(m)$ denotes a group of unitary $m \times m$ matrices. $\mu$ denotes a corresponding Haar measure. 
In addition, $C$, $C_1$~denote various $n$-independent constants which can be different in different formulas.

\section{Integral representation for $\CF_m$}\label{sec:IR}
The following integral representation is true
\begin{proposition}\label{prop:IR}
	Let an ensemble $\ens$ be defined by~\eqref{Gin ens} and~\eqref{moments}. Then 
	the $m$\textsuperscript{th} correlation function of the characteristic polynomials $\CF_m$ defined by \eqref{F_m} can be represented in the following form
	\begin{equation}\label{IR result}
	\CF_m = \left(\frac{n}{\pi}\right)^{c_m} \int g(\cSetMatA) e^{(n - c_m)f(\cSetMatA)} d\cSetMatA,
	\end{equation}
	where $c_m = 2^{2m - 1}$, $\cSetMatA = (\cSetMatA_j)_{j = 0}^m$, $\cSetMatA_j = \{\cMatA_{p,s} \mid p + s = 2j,\, 0 \le p,s \le m\}$, $\cMatA_{p,s}$ is a complex $\binom{m}{p} \times \binom{m}{s}$ matrix, $d\cSetMatA = \prod\limits_{\substack{p + s \text{ is even} \\ 0 \le p,s \le m}} d\cConjMat{\cMatA_{p,s}}d\cMatA_{p,s}$ and
	\begin{align}
		\label{f def}
		f(\cSetMatA) &= - \langle \cSetMatA, \cSetMatA \rangle + \log h(\cSetMatA); \\
		\notag
		g(\cSetMatA) &= (h(\cSetMatA)^{c_m} + n^{-1/2}\mathtt{p}_a(\cSetMatA) ) \exp\left\{-c_m\langle \cSetMatA, \cSetMatA \rangle \right\}; \\
		\label{h def}
		h(\cSetMatA) &= \Pf F(\cSetMatA_1) + n^{-1/2} \tilde{h}(\cSetMatA_1, \cSetMatA_2) + n^{-1}\mathtt{p}_c(\cSetMatA_1, \cSetMatA_{>1}); \\
		\label{A def}
		F(\cSetMatA_1) &= \begin{pmatrix}
			\sqrt{\seccumul} B_{2,0} & 0						  & -Z								  & \cMatA_1					\\
			0						 & \sqrt{\seccumul} \cConjMat{B_{0,2}} & -\cConjMat{\cMatA_1}		  & -\cConjMat{Z}					\\
			Z						 & \conj{\cMatA}_1			  & \sqrt{\conj{\seccumul}} \cConjMat{B_{2,0}} & 0					\\
			-\transp{\cMatA_1}		 & \cConjMat{Z}					  & 0								  & \sqrt{\conj{\seccumul}} B_{0,2}
		\end{pmatrix};
	\end{align}
	$B_{2,0}$ and $B_{0,2}$ are skew-symmetric matrices such that
	\begin{align*}
		\notag
		(B_{2,0})_{\alpha_1\alpha_2} = - \cSclA_{\alpha\emptyindex}^{(2,0)}, \quad (B_{0,2})_{\alpha_1\alpha_2} = - \cSclA_{\emptyindex\alpha}^{(0,2)}, \quad \alpha \in \indexset_{m,2}
	\end{align*}
	and $\indexset_{m, 2}$ is defined in \eqref{indexset def}. Moreover,
	\begin{gather}
		\label{tilde h def}
		\tilde{h}(\cSetMatA_1, \cSetMatA_2) = - \int \sum\limits_{p + s = 4} \left(\tr \tilde{\matA}_{p,s} \cMatA_{p,s} + \tr \cConjMat{\cMatA_{p,s}} \matA_{p,s}\right) e^{-\frac{1}{2}\transp{\aVecD}F\aVecD} d\aConjMat{\aVecBt} d\aVecBt d\aConjMat{\aVecCt} d\aVecCt, \\
		\label{rho def}
		\aVecD = 
		\transp{\begin{pmatrix}
		\aConjMat{\aVecB} & \aConjMat{\aVecC} & \transp{\aVecB} & \transp{\aVecC}
		\end{pmatrix}}, \\
		\label{chi def}
		\begin{split}
		\tilde{\sclA}_{\beta\alpha}^{(p,s)} &= \sqrt{\cumul{p}{s}}(-1)^p\prod\limits_{r = s}^{1} \aSclC_{\beta_{r}}^{\phantom{+}} \prod\limits_{q = p}^{1} \aConjScl{\aSclB_{\alpha_{q}}}, \\
		\sclA_{\alpha\beta}^{(p,s)} &= \sqrt{\cumul{p}{s}} \prod\limits_{q = 1}^{p} \aSclB_{\alpha_{q}} \prod\limits_{r = 1}^{s} \aConjScl{\aSclC_{\beta_{r}}},
		\end{split}
	\end{gather}
	$\cumul{p}{s}$ are defined in \eqref{cumul def}, $\mathtt{p}_a(\cSetMatA)$ and $\mathtt{p}_c(\cSetMatA_1, \cSetMatA_{>1})$ are certain polynomials such that $\mathtt{p}_c(\cSetMatA_1, 0) = 0$, 
	and $\cSetMatA_{>1}$ contains all $\cSetMatA_j$ except $\cSetMatA_1$.
\end{proposition}
\begin{proof}
Proposintion~\ref{prop:IR} was proved for the case $\seccumul = 1$ in~\cite[Proposition~2.1]{Af:20}. The most part of the provided proof goes in the frames of a general case, and only in the very end $\seccumul = 1$ is substituted. Therefore it is easy to understand from \cite{Af:20} that the only distinction of the general case from the partial one is in the presence of $\seccumul$ in \eqref{A def}.
\end{proof}
\begin{remark}
	Let $\cMatA_1 = U\Lambda \cConjMat{V}$ be the singular value decomposition of the matrix $\cMatA_1$, i.e.\ $\Lambda = \diag\{\lambda_j\}_{j = 1}^m$, $\lambda_j \ge 0$, $U, V \in U(m)$. In order to perform asymptotic analysis let us change the variables $\cMatA_1 = U\Lambda \cConjMat{V}$, $B_{2,0} \to UB_{2,0}\transp{U}$, $B_{0,2} \to \conj{V}B_{0,2}\cConjMat{V}$ in \eqref{IR result}. Since the Jacobian is $\frac{2^m\pi^{m^2}}{\left(\prod_{j = 1}^{m - 1} j!\right)^2}\Vanddet^2(\Lambda^2) \prod\limits_{j = 1}^m \lambda_j$ (see e.g.\ \cite{Hu:63}) we obtain
	\begin{equation}\label{IR SVD}
	\begin{split}
	\CF_m = Cn^{c_m} &\int\limits_\idom \Vanddet^2(\Lambda^2) \prod\limits_{j = 1}^m \lambda_j \left[ g_0(\Lambda, \hat{\cSetMatA}) + \frac{1}{\sqrt{n}}g_r(U\Lambda \cConjMat{V}, \hat{\cSetMatA}) \right] \\
	&\times \exp\left\{(n - c_m)\left[ f_0(\Lambda, \hat{\cSetMatA}) + \frac{1}{\sqrt{n}}f_r(U\Lambda \cConjMat{V}, \hat{\cSetMatA}) \right]\right\} \\
	&\times d\mu(U) d\mu(V) d\Lambda d\hat{\cSetMatA},
	\end{split}
	\end{equation}
	where $\hat{\cSetMatA}$ contains all the matrices $\cMatA_{p,s}$ except $\cMatA_1$, $\idom = \{(\Lambda, U, V, \hat{\cSetMatA}) \mid \lambda_j \ge 0,\, j = 1, \dotsc, m,\linebreak[0] U, V \in U(m)\}$, $\mu$ is a Haar measure, $d\Lambda = \prod\limits_{j = 1}^m d\lambda_j$ and
	\begin{align}
		\label{f_0 def}
		f_0(\cSetMatA) &= -\langle \cSetMatA, \cSetMatA \rangle + \log h_0(\cSetMatA_1); \\
		\notag
		g_0(\cSetMatA) &= h_0(\cSetMatA_1)^{c_m} \exp\left\{-c_m\langle \cSetMatA, \cSetMatA \rangle \right\} = e^{c_m f_0(\cSetMatA)}; \\
		\label{h_0 def}
		h_0(\cSetMatA_1) &= \Pf \tilde{F}(\cSetMatA_1),\,\, \tilde{F}(\cSetMatA_1) :=
		\begin{pmatrix}
			\skB_{2,0} & 0							 & -z_0I_m			& \Lambda			\\
			0				& \cConjMat{\skB_{0,2}} & -\Lambda			& -\conj{z}_0I_m	\\
			z_0I_m			& \Lambda			& \cConjMat{\skB_{2,0}} & 0				\\
			-\Lambda		& \conj{z}_0I_m		& 0					& \skB_{0,2}
		\end{pmatrix}; \\
		\label{f_r def}
		f_r(\cSetMatA) &= \sqrt{n}(f(\cSetMatA) - f_0(\cSetMatA)); \\
		\notag 
		g_r(\cSetMatA) &= \sqrt{n}(g(\cSetMatA) - g_0(\cSetMatA)),
	\end{align}
	$\skB_{2,0} = \sqrt{\seccumul}B_{2,0}$ and $\skB_{0,2} = \sqrt{\conj{\seccumul}}B_{0,2}$. Notice that $f_0(U\Lambda \cConjMat{V}, \hat{\cSetMatA}) = f_0(\Lambda, \hat{\cSetMatA})$ and the same for $g_0$.
\end{remark}
\begin{remark}\label{rem:case m = 1}
	In the special case $m = 1$ the matrices $B_{2,0}$ and $B_{0,2}$ are zeros and we have
	\begin{equation*} 
		\CF_1(z) = \frac{n}{\pi} \int\exp \left\{n(-\abs{q}^2+\log(\abs{z}^2+\abs{q}^2))\right\} d\bar{q}dq.
	\end{equation*}
	Changing variables to polar coordinates and performing a simple Laplace integration, we obtain
	\begin{equation}\label{F_1 behavior}
	\begin{split}
		\CF_1(z) &= 2n \int\limits_0^{+\infty} r\exp \left\{n(-r^2+\log(\abs{z}^2+r^2))\right\} dr \\
		&= \sqrt{2\pi n}\, e^{n(\abs{z}^2 - 1)}(1 + o(1)).
	\end{split}
	\end{equation}
\end{remark}

\section{Asymptotic analysis}\label{sec:asympt analysis}
The goal of the section is to investigate the asymptotic behavior of the integral representation~\eqref{IR SVD}. To this end, the steepest descent method is applied. As usual, the hardest step is to choose stationary points of $f(\cSetMatA)$ and a $N$-dimensional (real) manifold $M_* \subset \Compl^{N}$ such that for any chosen stationary point $\cSetMatA_* \in M_*$
\begin{equation*}
	\Re f(\cSetMatA) < \Re f(\cSetMatA_*), \quad \forall \cSetMatA \in M_*,\, \text{$\cSetMatA$ is not chosen}.
\end{equation*}
Note that $N$ is equal to the number of real variables of the integration, i.e.\ in our case $N = 2^{2m}$.

The present proof proceeds by a standard scheme for the case when function $f(\cSetMatA)$ has the form
\begin{equation*}
	f(\cSetMatA) = f_0(\cSetMatA) + n^{-1/2}f_r(\cSetMatA),
\end{equation*}
where $f_0(\cSetMatA)$ does not depend on $n$, whereas $f_r(\cSetMatA)$ may depend on $n$. We choose stationary points of $f_0(\cSetMatA)$ of the form $\cMatA_1 = U\Lambda_0\cConjMat{V}$, $\hat{\cSetMatA} = 0$, where $\Lambda_0 = 
\lambda_0I$, $\lambda_0 = \sqrt{1 - \abs{z_0}^2}$ 
and $U$, $V$ vary in $U(m)$. The manifold $M_*$ is $\R^N$. Then the steepest descent method is applied to the integral over $\Lambda$ and $\hat{\cSetMatA}$. In the process $U$ and $V$ are considered as parameters and all the estimates are uniform in $U$ and $V$. As soon as the domain of integration is restricted by a small neighborhood we recall about the integration over $U$ and $V$. After several changes of the variables the integral is reduced to the form \eqref{main result1}.

We start with an analysis of the function $f_0$. 
\begin{lemma}\label{lem:max of Re f_0}
	Let the function $f_0 \colon \R^{2^{2m}} \to \Compl$ be defined by \eqref{f_0 def}. Then the function $\Re f_0(\Lambda, \hat{\cSetMatA})$ attains its global maximum value only at the point
	\begin{equation}\label{point of max}
		\lambda_1 = \dotsb = \lambda_m = \lambda_0, \quad \hat{\cSetMatA} = 0,
	\end{equation}
	where $\lambda_0 = \sqrt{1 - \abs{z_0}^2}$. 
\end{lemma}
\begin{proof}
	From 
	\eqref{f_0 def} and \eqref{h_0 def} we get
	\begin{equation}\label{max of Re: ineq 1}
	\begin{split}
	\Re f_0(\Lambda, \hat{\cSetMatA}) &= - \sum\limits_{j \ne 1} \langle \cSetMatA_j, \cSetMatA_j \rangle - \langle \cSetMatA_1, \cSetMatA_1 \rangle + \frac{1}{2}\log \abs{\det \tilde{F}} \\
	&\le - \langle \cSetMatA_1, \cSetMatA_1 \rangle + \frac{1}{2}\log \abs{\det \tilde{F}}.
	\end{split}
	\end{equation}
	Hadamard's inequality yields
	\begin{equation}\label{max of Re: Hadamard}
	\begin{split}
	&
	\frac{1}{2}\log \abs{\det \tilde{F}} \le 
	\frac{1}{2}\log \biggl\lbrace \prod\limits_{j = 1}^{m} \Bigl(\abs{z_0}^2 + \lambda_j^2 + \abs{\seccumul} \sum\limits_{k = 1}^m \abs{\cSclA_{(j,k)\emptyindex}^{(2,0)}}^2 
	\Bigr)^{\frac{1}{2}} \\
	&\times \Bigl(\abs{\conj{z}_0}^2 + \lambda_j^2 + \abs{\seccumul} \sum\limits_{k = 1}^m \abs{\conj{\cSclA_{\emptyindex(j,k)}^{(0,2)}}}^2 \Bigr)^{\frac{1}{2}}\Bigl(\abs{z_0}^2 + \lambda_j^2 + \abs{\conj{\seccumul}} \sum\limits_{k = 1}^m \abs{\conj{\cSclA_{(j,k)\emptyindex}^{(2,0)}}}^2 \Bigr)^{\frac{1}{2}} \\
	&\phantom{\Bigl(\abs{z_0}^2 + \lambda_j^2 + \abs{\seccumul} \sum\limits_{k = 1}^m \abs{\cSclA_{(j,k)\emptyindex}^{(2,0)}}^2 \Bigr)^{\frac{1}{2}}}\times \Bigl(\abs{\conj{z}_0}^2 + \lambda_j^2 + \abs{\conj{\seccumul}} \sum\limits_{k = 1}^m \abs{\cSclA_{\emptyindex(j,k)}^{(0,2)}}^2 \Bigr)^{\frac{1}{2}} \biggr\rbrace.
	\end{split}
	\end{equation}
	where $\cSclA_{(j,k)\emptyindex}^{(2,0)} = - \cSclA_{(k,j)\emptyindex}^{(2,0)}$, $\cSclA_{(j,k)\emptyindex}^{(0,2)} = - \cSclA_{(k,j)\emptyindex}^{(0,2)}$ for $j > k$ and $\cSclA_{(j,j)\emptyindex}^{(2,0)} = \cSclA_{(j,j)\emptyindex}^{(0,2)} = 0$. Simplifying the r.h.s. of \eqref{max of Re: Hadamard} and taking into account \eqref{max of Re: ineq 1} we obtain
	\begin{equation}\label{max of Re: ineq 2}
	\begin{split}
	\Re f_0(\Lambda, \hat{\cSetMatA}) \le - \langle \cSetMatA_1, \cSetMatA_1 \rangle + \frac{1}{2}\sum\limits_{j = 1}^{m} \log \biggl\lbrace \Bigl(\abs{z_0}^2 + \lambda_j^2 + \abs{\seccumul} \sum\limits_{k = 1}^m \abs{\cSclA_{(j,k)\emptyindex}^{(2,0)}}^2 \Bigr) \\
	\times \Bigl(\abs{z_0}^2 + \lambda_j^2 + \abs{\seccumul} \sum\limits_{k = 1}^m \abs{\cSclA_{\emptyindex(j,k)}^{(0,2)}}^2 \Bigr) \biggr\rbrace,
	\end{split}
	\end{equation}
	The inequality $\log x \le x - 1$ and \eqref{max of Re: ineq 2} imply
	\begin{equation}\label{max of Re: ineq 3}
	\begin{split}
	\Re f_0(\Lambda, \hat{\cSetMatA}) \le{}& {-\langle \cSetMatA_1, \cSetMatA_1 \rangle} + \frac{1}{2}\sum\limits_{j = 1}^{m} \biggl\lbrace \Bigl(\abs{z_0}^2 + \lambda_j^2 + \abs{\seccumul} \sum\limits_{k = 1}^m \abs{\cSclA_{(j,k)\emptyindex}^{(2,0)}}^2 \Bigr) \\
	&+ \Bigl(\abs{z_0}^2 + \lambda_j^2 + \abs{\seccumul} \sum\limits_{k = 1}^m \abs{\cSclA_{\emptyindex(j,k)}^{(0,2)}}^2 \Bigr) - 2 \biggr\rbrace \\
	={}& {-\langle \cSetMatA_1, \cSetMatA_1 \rangle} + m\abs{z_0}^2 - m + \sum\limits_{j = 1}^{m} \lambda_j^2 + \abs{\seccumul} \sum\limits_{\alpha \in \indexset_{m,2}} \abs{\cSclA_{\alpha\emptyindex}^{(2,0)}}^2 \\
	&+ \abs{\seccumul} \sum\limits_{\alpha \in \indexset_{m,2}} \abs{\cSclA_{\emptyindex\alpha}^{(0,2)}}^2.
	\end{split}
	\end{equation}
	Finally, since $\abs{\seccumul} = \abs{\E\lbrace x_{11}^2 \rbrace} \le \E\lbrace \abs{x_{11}}^2 \rbrace = 1$, we have
	\begin{multline}\label{max of Re: main ineq}
	\Re f_0(\Lambda, \hat{\cSetMatA}) \le {-\langle \cSetMatA_1, \cSetMatA_1 \rangle} + \sum\limits_{j = 1}^{m} \lambda_j^2 + \sum\limits_{\alpha \in \indexset_{m,2}} \Bigl\lbrace
	\abs{\cSclA_{\alpha\emptyindex}^{(2,0)}}^2 + \abs{\cSclA_{\emptyindex\alpha}^{(0,2)}}^2 \Bigr\rbrace \\
	+ m\abs{z_0}^2 - m = {-\langle \cSetMatA_1, \cSetMatA_1 \rangle} + \langle \cSetMatA_1, \cSetMatA_1 \rangle + m\abs{z_0}^2 - m = m(\abs{z_0}^2 - 1).
	\end{multline}
	Therefore, the function $\Re f_0(\Lambda, \hat{\cSetMatA})$ attains its global maximum value at the point~\eqref{point of max}. It remains to show that there is no other point for which $\Re f_0(\Lambda, \hat{\cSetMatA}) = m(\abs{z_0}^2 - 1)$. Indeed, equality in \eqref{max of Re: ineq 1} is attained if and only if $\cSetMatA_{>1} = 0$. Moreover, 
	the r.h.s. of \eqref{max of Re: main ineq} and \eqref{max of Re: ineq 3} are equal 
	if and only if $\abs{\seccumul} = 1$ or $\cSclA_{\alpha\emptyindex}^{(2,0)} = \cSclA_{\emptyindex\alpha}^{(0,2)} = 0$. Let us consider the following two cases.
	\begin{enumerate}
	\item $\abs{\seccumul} < 1 \quad \conseq \quad \cSclA_{\alpha\emptyindex}^{(2,0)} = \cSclA_{\emptyindex\alpha}^{(0,2)} = 0$ for all $\alpha \in \indexset_{m,2}$.
	
	Since the equality $\log x = x - 1$ holds if and only if $x = 1$, then we obtain from the equality of the r.h.s.\ of \eqref{max of Re: ineq 2} and \eqref{max of Re: ineq 3} that
	\begin{equation*}
	\abs{z_0}^2 + \lambda_j^2 + \abs{\seccumul} \sum\limits_{k = 1}^m \abs{\cSclA_{(j,k)\emptyindex}^{(2,0)}}^2 = 1.
	\end{equation*}
	Thus for any $j$
	\begin{equation*}
	\lambda_j = \sqrt{1 - \abs{z_0}^2}.
	\end{equation*}
	\item $\abs{\seccumul} = 1$ and $z_0 \notin \R$.
	
	Equality in Hadamard's inequality is attained if and only if columns of a matrix are orthogonal vectors. Hence, if equality is attained in \eqref{max of Re: Hadamard}, then the columns of the matrix~$\tilde{F}$ are orthogonal. In particular, the orthogonality of the first and the $2m + 2$\textsuperscript{nd} yields
	\begin{equation*}\label{1 perp 2m+2}
	-(\skB_{2,0})_{21}\conj{z}_0 + z_0(\skB_{2,0})_{21} = 0.
	\end{equation*}
	Since $z_0 \ne \conj{z}_0$, the last identity implies
	\begin{equation*}
	q_{(1,2)\emptyindex}^{(2,0)} = \frac{1}{\sqrt{\seccumul}}(\skB_{2,0})_{21} = 0.
	\end{equation*}
	Using a similar argument, we get that all $q_{\alpha\emptyindex}^{(2,0)}$ and $q_{\emptyindex\alpha}^{(0,2)}$ are zeros. Next, similarly to the first case we obtain $\lambda_1 = \dotsb = \lambda_m = \sqrt{1 - \abs{z_0}^2}$.
	\end{enumerate}
	Totally, the assertion of the lemma is proven.
\end{proof}

To simplify the reading, the remaining steps are first explained in the case when the cumulants $\cumul{p}{s}$, $p + s > 2$ are zeros.

\subsection{Case of zero high cumulants}
Now we proceed to the integral estimates. In a standard way the integration domain in \eqref{IR SVD} can be restricted as follows
\begin{align*}
	\CF_m = Cn^{c_m} \int\limits_{\Sigma_r} \Vanddet^2(\Lambda^2) \prod\limits_{j = 1}^m \lambda_j \times g(U\Lambda \cConjMat{V}, \hat{\cSetMatA})e^{(n - c_m)f(U\Lambda \cConjMat{V}, \hat{\cSetMatA})} d\mu(U) d\mu(V) d\Lambda d\hat{\cSetMatA} \\
	+ O(e^{-nr/2}),
\end{align*}
where
\begin{equation*}
	\Sigma_r = \left\{(\Lambda, U, V, \hat{\cSetMatA}) \mid \norm{\Lambda} + \normsized{\hat{\cSetMatA}} \le r\right\} 
	.
\end{equation*}
The next step is to restrict the integration domain by 
\begin{equation}\label{stpoinnbh def}
\stpointsnbh = \left\{(\Lambda, U, V, \hat{\cSetMatA}) \mid \norm{\Lambda - \Lambda_0} + \normsized{\hat{\cSetMatA}} \le \frac{\log n}{\sqrt{n}}\right\}.
\end{equation}
To this end we need the estimate of $\Re f$ given by the following lemmas.

\begin{lemma}\label{lem:f(UL V^*) expansion}
	Let $\tilde{\Lambda}$ and $\hat{\tilde{\cSetMatA}}$ satisfy the condition $\normsized{\tilde{\Lambda}} + \normsized{\hat{\tilde{\cSetMatA}}} \le \log n$. Then uniformly in $U$ and $V$
	\begin{multline} \label{f expansion}
	f(U(\Lambda_0 + n^{-1/2}\tilde{\Lambda})\cConjMat{V}, n^{-1/2}\hat{\tilde{\cSetMatA}}) = {- m\lambda_0^2} + n^{-1/2} \tr (\conj{z}_0\cMatPos + z_0\cConjMat{\cMatPos}) \\
	\shoveright{- \frac{1}{2n} \tr(2\lambda_0\tilde{\Lambda} + \conj{z}_0\cMatPos_U + z_0\cConjMat{\cMatPos_V})^2 + \frac{1}{n} \tr\cMatPos_U\cConjMat{\cMatPos_V}} \\
	\shoveright{- \frac{1}{2n}\tr\Bigl[(1 - \abs{\seccumul}\conj{z}_0^2)\cConjMat{\tilde{B}_{2,0}}\tilde{B}_{2,0} + (1 - \abs{\seccumul}z_0^2)\cConjMat{\tilde{B}_{0,2}}\tilde{B}_{0,2}} \\
	\shoveright{- \abs{\seccumul}\lambda_0^2\tilde{B}_{0,2}\tilde{B}_{2,0} - \abs{\seccumul}\lambda_0^2\cConjMat{\tilde{B}_{2,0}}\cConjMat{\tilde{B}_{0,2}}\Bigr] - \frac{1}{n}\norm{\tilde{\cSetMatA}_{>1}}^2} \\
	+ O\big(n^{-3/2}\log^3 n\big)\lefteqn{,}
	\end{multline}
	where $\cMatPos_W = \cConjMat{W}\cMatPos W$.
\end{lemma}

\begin{proof}
	If $Q_1 = U(\Lambda_0 + n^{-1/2}\tilde{\Lambda})\cConjMat{V}$, then $F$ has the form
	\begin{equation*}
		F = \begin{pmatrix}
			U	& 0	& 0			& 0			\\
			0	& V	& 0			& 0			\\
			0	& 0	& \conj{U}	& 0			\\
			0	& 0	& 0			& \conj{V}
		\end{pmatrix}\left(F_0 + \frac{1}{\sqrt{n}}F_1\right)\begin{pmatrix}
		\transp{U}	& 0			 & 0			& 0			\\
		0			& \transp{V} & 0			& 0			\\
		0			& 0			 & \cConjMat{U}	& 0			\\
		0			& 0			 & 0			& \cConjMat{V}
	\end{pmatrix},
\end{equation*}
where
\begin{equation}\label{F_0,F_1 def}
\begin{gathered}
F_0 = \begin{pmatrix}
0 			& A_0	\\
-\transp{A_0}	& 0
\end{pmatrix}, \quad
F_1 = \begin{pmatrix}
B			& A				\\
-\transp{A}	& \cConjMat{B}
\end{pmatrix}, \\
A_0 = \begin{pmatrix}
-z_0I_m		& \Lambda_0			\\
-\Lambda_0	& -\conj{z}_0I_m
\end{pmatrix}, \quad
A = \begin{pmatrix}
-\cMatPos_{U}		& \tilde{\Lambda}			\\
-\tilde{\Lambda}	& -\cConjMat{\cMatPos_{V}}
\end{pmatrix}, \quad
B = \begin{pmatrix}
\tilde{\skB}_{2,0} & 0							 \\
0				& \cConjMat{\tilde{\skB}_{0,2}}
\end{pmatrix}.
\end{gathered}
\end{equation}
Taking into account that
\begin{equation*}
	\det F_0 = \left[\det 
	\begin{pmatrix}
		z_0		& \lambda_0 \\
		-\lambda_0	& \conj{z}_0
	\end{pmatrix}\det 
	\begin{pmatrix}
		-z_0		& \lambda_0 \\
		-\lambda_0	& -\conj{z}_0
	\end{pmatrix}\right]^m
	= 1,
\end{equation*}
one gets
\begin{equation}\label{log det F}
\begin{split}
\log \det F &= \tr \log (1 + n^{-1/2}F_0^{-1}F_1) \\
&= \frac{1}{\sqrt{n}}\tr F_0^{-1}F_1 - \frac{1}{2n}\tr (F_0^{-1}F_1)^2 + O\left(\frac{\log^3 n}{\sqrt{n^3}}\right)
\end{split}
\end{equation}
uniformly in $U$ and $V$. Further,
\begin{equation}\label{F_0^-1F_1}
F_0^{-1}F_1 = \begin{pmatrix}
\bigl(\transp{A_0}\bigr)^{-1}\transp{A}	& -\bigl(\transp{A_0}\bigr)^{-1}\cConjMat{B} \\
A_0^{-1}B & A_0^{-1}A
\end{pmatrix}
\end{equation}
and
\begin{multline}\label{(F_0^-1F_1)^2}
(F_0^{-1}F_1)^2 = \\
\begin{pmatrix}
\Bigl(\transp{\bigl(A_0^{-1}A\bigr)}\Bigr)^2 - \bigl(\transp{A_0}\bigr)^{-1}\cConjMat{B}A_0^{-1}B	& * \\
* & -A_0^{-1}B\bigl(\transp{A_0}\bigr)^{-1}\cConjMat{B} + \left(A_0^{-1}A\right)^2
\end{pmatrix}.
\end{multline}
Moreover,
\begin{equation}\label{A,B calc}
\begin{split}
A_0^{-1}A &= \begin{pmatrix}
\bar{z}_0\cMatPos_U + \lambda_0\tilde{\Lambda}	& -\bar{z}_0\tilde{\Lambda} + \lambda_0\cMatPos_\cConjMat{V} \\
-\lambda_0\cMatPos_U + z_0\tilde{\Lambda}		& \lambda_0\tilde{\Lambda} + z_0\cMatPos_\cConjMat{V}
\end{pmatrix}, \\
\bigl(\transp{A_0}\bigr)^{-1}\cConjMat{B}A_0^{-1}B &= \begin{pmatrix}
\conj{z}_0^2\cConjMat{\tilde{\skB}_{2,0}}\tilde{\skB}_{2,0} + \lambda_0^2\tilde{\skB}_{0,2}\tilde{\skB}_{2,0}	& * \\
* & \lambda_0^2\cConjMat{\tilde{\skB}_{2,0}}\cConjMat{\tilde{\skB}_{0,2}} + \conj{z}_0^2\tilde{\skB}_{0,2}\cConjMat{\tilde{\skB}_{0,2}}
\end{pmatrix}\ldotp
\end{split}
\end{equation}
Combining \eqref{log det F}--\eqref{A,B calc} and \eqref{f def}, we get
\begin{multline*}\label{last expansion}
	f(U(\Lambda_0 + n^{-1/2}\tilde{\Lambda})\cConjMat{V}, n^{-1/2}\hat{\tilde{\cSetMatA}}) = -\tr\Bigl[\Lambda_0^2 + 2n^{-1/2}\lambda_0\tilde{\Lambda} + n^{-1}\tilde{\Lambda}^2\Bigr] \\
	\shoveright{- \frac{1}{2n}\tr[\cConjMat{\tilde{B}_{2,0}}\tilde{B}_{2,0} + \cConjMat{\tilde{B}_{0,2}}\tilde{B}_{0,2}] - \frac{1}{n} 
	\norm{\tilde{\cSetMatA}_{>1}}^2 + \frac{1}{n^{1/2}}\tr[2\lambda_0 \tilde{\Lambda} + \bar{z}_0 \cMatPos_U + z_0 \cMatPos_\cConjMat{V}]}\\
	\shoveright{- \frac{1}{n}\tr\Bigl[(\lambda_0^2 - \abs{z_0}^2) \tilde{\Lambda}^2 + 2\bar{z}_0\lambda_0\cMatPos_U\tilde{\Lambda}+ 2z_0\lambda_0\cMatPos_\cConjMat{V}\tilde{\Lambda} + \frac{1}{2}(\bar{z}_0\cMatPos_U + z_0\cMatPos_\cConjMat{V})^2 - \cMatPos_U\cMatPos_\cConjMat{V}\Bigr]} \\
	\shoveright{+\frac{1}{2n}\abs{\seccumul}\tr[\conj{z}_0^2\cConjMat{\tilde{B}_{2,0}}\tilde{B}_{2,0} + \lambda_0^2\tilde{B}_{0,2}\tilde{B}_{2,0} + \lambda_0^2\cConjMat{\tilde{B}_{2,0}}\cConjMat{\tilde{B}_{0,2}} + \conj{z}_0^2\tilde{B}_{0,2}\cConjMat{\tilde{B}_{0,2}}]} \\
	+ O\big(n^{-3/2}\log^3 n\big)\lefteqn{\ldotp}
\end{multline*}
Hence the last expansion yields \eqref{f expansion}.
\end{proof}

\begin{corollary}\label{cor:deriv of f_0}
Let the function $f_0 \colon \R^{2^{2m}} \to \Compl$ be defined by \eqref{f_0 def}. Then the following assertions are true:
\begin{enumerate}[label=(\roman*)]
\item\label{st point} the point $\cSetMatA_* = (\Lambda_0, 0)$ is a stationary point of the function $f_0(\Lambda, \hat{\cSetMatA})$;
\item\label{Hess} the Hessian matrix 
of the function $\Re f_0(\Lambda, \hat{\cSetMatA})$ (as a function of real argument) 
at the point $\cSetMatA_* 
$ is negative definite.
\end{enumerate}
\end{corollary}

\begin{proof}
Let us put $\cMatPos = 0$ and $\mathtt{p}_c = 0$. Then
\begin{equation*}
f_0(\Lambda, \hat{\cSetMatA}) = f(\Lambda, \hat{\cSetMatA}).
\end{equation*}
Therefore it is possible to consider the expansion~\eqref{f expansion} as Taylor formula for $f_0(\Lambda, \hat{\cSetMatA})$ at the point $(\Lambda_0, 0)$. We obtain
\begin{multline*} 
f_0(\Lambda_0 + n^{-\frac{1}{2}}\tilde{\Lambda}, n^{-\frac{1}{2}}\hat{\tilde{\cSetMatA}}) = {- m\lambda_0^2} - n^{-1} 2\lambda_0^2 \tr \tilde{\Lambda}^2 - n^{-1}\norm{\tilde{\cSetMatA}_{>1}}^2 \\
\shoveright{- \frac{1}{2n}\tr\Bigl\lbrack(1 - \abs{\seccumul}\conj{z}_0^2)\cConjMat{\tilde{B}_{2,0}}\tilde{B}_{2,0} + (1 - \abs{\seccumul}z_0^2)\cConjMat{\tilde{B}_{0,2}}\tilde{B}_{0,2}} \\
- \abs{\seccumul}\lambda_0^2\tilde{B}_{0,2}\tilde{B}_{2,0} - \abs{\seccumul}\lambda_0^2\cConjMat{\tilde{B}_{2,0}}\cConjMat{\tilde{B}_{0,2}}\Bigr\rbrack + O\big(n^{-\frac{3}{2}}\log^3 n\big)\lefteqn{.}
\end{multline*}
Thus the gradient of the function $f_0(\Lambda, \hat{\cSetMatA})$ is evidently zero at the point $(\Lambda_0, 0)$. Assertion \ref{st point} is proven. Note that
\begin{multline}\label{B:variables separation}
\frac{1}{2n}\tr\Bigl\lbrack(1 - \abs{\seccumul}\conj{z}_0^2)\cConjMat{\tilde{B}_{2,0}}\tilde{B}_{2,0} + (1 - \abs{\seccumul}z_0^2)\cConjMat{\tilde{B}_{0,2}}\tilde{B}_{0,2} \\
\shoveleft{- \abs{\seccumul}\lambda_0^2\tilde{B}_{0,2}\tilde{B}_{2,0} - \abs{\seccumul}\lambda_0^2\cConjMat{\tilde{B}_{2,0}}\cConjMat{\tilde{B}_{0,2}} \Bigr\rbrack} \\
\shoveright{= \frac{1}{n} \sum\limits_{\alpha \in \indexset_{m,2}} \biggl\lbrack(1 - \abs{\seccumul}\conj{z}_0^2)\abs{\cSclA_{\alpha\emptyindex}^{(2,0)}}^2 + (1 - \abs{\seccumul}z_0^2)\abs{\cSclA_{\emptyindex\alpha}^{(0,2)}}^2} \\
+ \abs{\seccumul}\lambda_0^2 \left(\cSclA_{\alpha\emptyindex}^{(2,0)} \cSclA_{\emptyindex\alpha}^{(0,2)} + \conj{\cSclA_{\alpha\emptyindex}^{(2,0)} \cSclA_{\emptyindex\alpha}^{(0,2)}}\right) \biggr\rbrack\lefteqn{.}
\end{multline}
Hence, in order to prove assertion \ref{Hess} it is enough to show that the following quadratic form of $x_1$ and $x_2$
\begin{equation*}
(1 - \abs{\seccumul}\Re z_0^2)x_1^2 + (1 - \abs{\seccumul}\Re z_0^2)x_2^2 \pm 2\abs{\seccumul}\lambda_0^2 x_1 x_2
\end{equation*}
is positive definite. A straightforward check yields
\begin{align}
\notag
1 - \abs{\seccumul}\Re z_0^2 &> 0; \\
\label{positive det}
(1 - \abs{\seccumul}\Re z_0^2)^2 - \abs{\seccumul}^2\lambda_0^4 &\ge (1 - \abs{\seccumul z_0^2})^2 - \abs{\seccumul}^2\lambda_0^4 \ge 0.
\end{align}
Besides, if parameters $\seccumul$ and $z_0$ are such those in the assertion of Theorem~\ref{th1} then the inequality~\eqref{positive det} is strict.
\end{proof}

\begin{lemma}\label{lem:est for Re f}
	Let $\tilde{f}(\cMatA_1, \hat{\cSetMatA}) = f(\cMatA_1, \hat{\cSetMatA}) - f(\Lambda_0, 0)$. Then for sufficiently large~$n$
	\begin{equation*} 
		\max_{\frac{\log n}{\sqrt{n}} \le \norm{\Lambda - \Lambda_0} + \norm{\hat{\cSetMatA}} \le r} \Re \tilde{f}(U\Lambda\cConjMat{V}, \hat{\cSetMatA}) \le -C\frac{\log^2 n}{n}
	\end{equation*}
	uniformly in $U$ and $V$.
\end{lemma}
\begin{proof}
	First let us check that the first and the second derivatives of $f_r$ are bounded in the $\delta$-neighborhood of $\Lambda_0$, where $f_r$ is defined in \eqref{f_r def} and $\delta$ is $n$-independent. Indeed, since $h$ and $h_0$ are polynomials and $h \rightrightarrows h_0$ on compacts
	\begin{align*}
		\abs{\frac{1}{\sqrt{n}}\frac{\partial \Re f_r}{\partial x}} &\le \abs{\frac{1}{\sqrt{n}}\frac{\partial f_r}{\partial x}} = \abs{\frac{\partial (f - f_0)}{\partial x}} = \abs{\frac{\partial (\log h - \log h_0)}{\partial x}} \\
		&\le \abs{\frac{1}{h_0} \cdot \frac{\partial h_0}{\partial x} - \frac{1}{h} \cdot \frac{\partial h}{\partial x}} \le \frac{C}{\sqrt{n}},
	\end{align*}
	where $x$ is either $\lambda_j$ or an entry of $\cMatA_{p,s}$, $(p,s) \ne (1,1)$. Let $\Lambda_E$ be a real diagonal matrix of unit norm and let $\hat{\cSetMatA}_E$, $\normsized{\hat{\cSetMatA}_E} = 1$, be a set of matrices which sizes correspond to those of $\hat{\cSetMatA}$. Then for any $\Lambda_E$ and $\hat{\cSetMatA}_E$ and for $\frac{\log n}{\sqrt{n}} \le t \le \delta$ we have
	\begin{equation*}
		\begin{split}
			\der{}{t} \Re\tilde{f}(U(\Lambda_0 + t\Lambda_E)\cConjMat{V}, t\hat{\cSetMatA}_E) &= \langle \nabla_{\Lambda, \hat{\cSetMatA}} \Re f_0(U(\Lambda_0 + t\Lambda_E)\cConjMat{V}, t\hat{\cSetMatA}_E), v(E) \rangle \\
			&\quad{}+ n^{-1/2} \langle \nabla_{\Lambda, \hat{\cSetMatA}} \Re f_r(U(\Lambda_0 + t\Lambda_E)\cConjMat{V}, t\hat{\cSetMatA}_E), v(E) \rangle \\
			&= \langle \nabla_{\Lambda, \hat{\cSetMatA}} \Re f_0(\Lambda_0 + t\Lambda_E, t\hat{\cSetMatA}_E), v(E) \rangle + O(n^{-1/2}),
		\end{split}
	\end{equation*}
	where $v(E)$ denotes a vector witch components are all the real variables of $\Lambda_E$ and $\hat{\cSetMatA}_E$ and $\langle \cdot, \cdot \rangle$ is a standard real scalar product. Expanding the scalar product by Taylor formula and considering that $\nabla_{\Lambda, \hat{\cSetMatA}} f_0(\Lambda_0, 0) = 0$, we obtain
	\begin{equation*}
		\begin{split}
			\der{}{t} \Re\tilde{f}(U(\Lambda_0 + t\Lambda_E)\cConjMat{V}, t\hat{\cSetMatA}_E) &= t\langle (\Re f_0)''(\Lambda_0, 0)v(E), v(E) \rangle + r_1 + O(n^{-1/2}),
		\end{split}
	\end{equation*}
	where $(\Re f_0)''$ is a matrix of second order derivatives of $\Re f_0$ w.r.t.\ $\Lambda$, $\Re\hat{\cSetMatA}$ and $\Im\hat{\cSetMatA}$ and $\abs{r_1} \le Ct^2$. $(\Re f_0)''(\Lambda_0, 0)$ is negative definite according to Corollary~\ref{cor:deriv of f_0}. Hence $\der{}{t} \Re\tilde{f}(U(\Lambda_0 + t\Lambda_E)\cConjMat{V}, t\hat{\cSetMatA}_E)$ is negative and
	\begin{equation}\label{Re f: nbh est}
	\begin{split}
	\max_{\frac{\log n}{\sqrt{n}} \le \norm{\Lambda - \Lambda_0} + \norm{\hat{\cSetMatA}} \le \delta} \Re \tilde{f}(U\Lambda\cConjMat{V}, \hat{\cSetMatA}) &= \max_{\norm{\Lambda - \Lambda_0} + \norm{\hat{\cSetMatA}} = \frac{\log n}{\sqrt{n}}} \Re \tilde{f}(U\Lambda\cConjMat{V}, \hat{\cSetMatA}) \\
	&\le \Re f(U\Lambda_0 \cConjMat{V}, 0) - C\frac{\log^2 n}{n} - f(\Lambda_0, 0).
	\end{split}
	\end{equation}
	Notice that $f_r$ is bounded from above uniformly in $n$. This fact and Lemma \ref{lem:max of Re f_0} imply that $\delta$ in \eqref{Re f: nbh est} can be replaced by $r$
	\begin{equation*}
		\max_{\frac{\log n}{\sqrt{n}} \le \norm{\Lambda - \Lambda_0} + \norm{\hat{\cSetMatA}} \le r} \Re \tilde{f}(U\Lambda\cConjMat{V}, \hat{\cSetMatA}) \le \Re f(U\Lambda_0 \cConjMat{V}, 0) - f(\Lambda_0, 0) - C\frac{\log^2 n}{n}.
	\end{equation*}
	It remains to deduce from Lemma \ref{lem:f(UL V^*) expansion} that $\Re f(U\Lambda_0 \cConjMat{V}, 0) - f(\Lambda_0, 0) = O(n^{-1})$ uniformly in $U$ and $V$.
\end{proof}

Lemma \ref{lem:est for Re f} and the formula~\eqref{IR SVD} yield
\begin{multline*}
	\CF_m = Cn^{c_m}e^{nf(\Lambda_0, 0)} \Bigg(\int\limits_{\stpointsnbh} \Vanddet^2(\Lambda^2) \prod\limits_{j = 1}^m \lambda_j g(\cSetMatA) e^{-c_m f(U\Lambda \cConjMat{V}, \hat{\cSetMatA})} \\
	\times e^{n\tilde{f}(U\Lambda \cConjMat{V}, \hat{\cSetMatA})} d\mu(U) d\mu(V) d\Lambda d\hat{\cSetMatA} + O(e^{-C_1\log^2 n})\Bigg),
\end{multline*}
where $\stpointsnbh$ is defined in \eqref{stpoinnbh def}.
Changing the variables $\Lambda = \Lambda_0 + \frac{1}{\sqrt{n}}\tilde{\Lambda}$, $\hat{\cSetMatA} = \frac{1}{\sqrt{n}}\hat{\tilde{\cSetMatA}}$ and expanding $f$ according to Lemma~\ref{lem:f(UL V^*) expansion} we obtain
\begin{equation}\label{n-indep int}
\begin{split}
\CF_m &= C\mathsf{k}_n \int\limits_{\sqrt{n}\stpointsnbh} \Vanddet^2(\tilde{\Lambda}) g(\cSetMatA_*) e^{-c_m f(\Lambda_0, 0)} d\mu(U) d\mu(V) d\tilde{\Lambda} d\hat{\tilde{\cSetMatA}} (1 + o(1)) \\
&\quad \times \exp\Bigl\lbrace -\frac{1}{2}\tr(2\lambda_0\tilde{\Lambda} + \bar{z}_0\cMatPos_U + z_0\cMatPos_\cConjMat{V})^2 + \tr\cMatPos_U\cMatPos_\cConjMat{V} - \norm{\tilde{\cSetMatA}_{>1}}^2 \\
&\quad\phantom{\times \exp\Bigl\lbrace} - \frac{1}{2}\tr\Bigl[(1 - \abs{\seccumul}\conj{z}_0^2)\cConjMat{\tilde{B}_{2,0}}\tilde{B}_{2,0} + (1 - \abs{\seccumul}z_0^2)\cConjMat{\tilde{B}_{0,2}}\tilde{B}_{0,2} \\
&\quad\phantom{\times \exp\Bigl\lbrace - \frac{1}{2}\tr\Bigl[} - \abs{\seccumul}\lambda_0^2\tilde{B}_{0,2}\tilde{B}_{2,0} - \abs{\seccumul}\lambda_0^2\cConjMat{\tilde{B}_{2,0}}\cConjMat{\tilde{B}_{0,2}}\Bigr]\Bigr\rbrace,
\end{split}
\end{equation}
where
\begin{equation}\label{K_n def}
\mathsf{k}_n = n^{m^2/2}e^{-mn\lambda_0^2 + \sqrt{n}\tr \left(\bar{z}_0\cMatPos + z_0\cConjMat{\cMatPos}\right)}.
\end{equation}
Since \eqref{B:variables separation} the integral over $\hat{\tilde{\cSetMatA}}$ can be computed separately over real and imaginary parts of the entries of $\hat{\tilde{\cSetMatA}}$. Because $g(\cSetMatA_*) e^{-c_m f(\Lambda_0, 0)} = 1 + o(1)$, the integration implies
\begin{equation}\label{n-indep int 1}
\begin{split}
\CF_m &= C\mathsf{k}_n d_1(\seccumul)^{-m} \int \Vanddet^2(\tilde{\Lambda}) 
d\mu(U) d\mu(V) d\tilde{\Lambda} (1 + o(1)) \\
&\qquad\qquad\quad \times \exp\Bigl\lbrace -\frac{1}{2}\tr(2\lambda_0\tilde{\Lambda} + \bar{z}_0\cMatPos_U + z_0\cMatPos_\cConjMat{V})^2 + \tr\cMatPos_U\cMatPos_\cConjMat{V} \Bigr\rbrace,
\end{split}
\end{equation}
where
\begin{equation}\label{d1 def}
d_1(\seccumul) = \abs{1 - \abs{\seccumul}z_0^2}^2 - \abs{\seccumul}^2\lambda_0^4.
\end{equation}

Let us change the variables $V = WU$. Taking into account that the Haar measure is invariant w.r.t.\ shifts 
we get
\begin{equation*}
	\begin{split}
	\CF_m &= C\mathsf{k}_n d_1(\seccumul)^{-m} \int\limits_{\R^m} \int\limits_{U(m)} \int\limits_{U(m)}\Vanddet^2(\tilde{\Lambda}) d\mu(U) d\mu(W) d\tilde{\Lambda}(1 + o(1)) \\
	&\quad \times \exp\left\{-\frac{1}{2}\tr(2\lambda_0\tilde{\Lambda} + \cConjMat{U}(\bar{z}_0\cMatPos + z_0\cMatPos_\cConjMat{W})U)^2 + \tr\cMatPos \cConjMat{W}\cConjMat{\cMatPos} W\right\} \\
	&= C\mathsf{k}_n d_1(\seccumul)^{-m} \int\limits_{\R^m} \int\limits_{U(m)} \int\limits_{U(m)} \Vanddet^2(\tilde{\Lambda}) d\mu(U) d\mu(W) d\tilde{\Lambda}(1 + o(1)) \\
	&\quad \times \exp\left\{-\frac{1}{2}\tr(2\lambda_0U\tilde{\Lambda}\cConjMat{U} + (\bar{z}_0\cMatPos + z_0\cMatPos_\cConjMat{W}))^2 + \tr\cMatPos \cConjMat{W}\cConjMat{\cMatPos} W\right\}.
	\end{split}
\end{equation*}
The next step is to change the variables $H = U\tilde{\Lambda}\cConjMat{U}$. The Jacobian is $\frac{\prod_{j = 1}^{m - 1} j!}{(2\pi)^{m(m - 1)/2}}\Vanddet^{-2}(\tilde{\Lambda})$ (see e.g.\ \cite{Hu:63}). Thus
\begin{equation*}
	\begin{split}
		\CF_m &= C\mathsf{k}_n d_1(\seccumul)^{-m} \int\limits_{\herm_m} \int\limits_{U(m)}  d\mu(W) dH(1 + o(1)) \\
		&\qquad{}\times \exp\left\{-\frac{1}{2}\tr(2\lambda_0H + (\bar{z}_0\cMatPos + z_0\cMatPos_\cConjMat{W}))^2 + \tr\cMatPos \cConjMat{W}\cConjMat{\cMatPos} W\right\},
	\end{split}
\end{equation*}
where $\herm_m$ is a space of hermitian $m \times m$ matrices and
\begin{equation*}
	dH = \prod\limits_{j = 1}^m d(H)_{jj} \prod\limits_{j < k} d\Re (H)_{jk} d\Im (H)_{jk}.
\end{equation*}
The Gaussian integration over $H$ implies
\begin{equation}\label{last asympt}
\CF_m = C\mathsf{k}_n d_1(\seccumul)^{-m} \int\limits_{U(m)} \exp\left\{ \tr\cMatPos \cConjMat{W}\cConjMat{\cMatPos} W\right\} d\mu(W) (1 + o(1)).
\end{equation}

For computing the integral over the unitary group, the following Harish-Chan\-dra/It\-syk\-son--Zuber formula is used

\begin{proposition}\label{pr:H-C/I--Z formula}
	Let $A$ and $B$ be normal $d \times d$ matrices with distinct eigenvalues $\{a_j\}_{j = 1}^d$ and $\{b_j\}_{j = 1}^d$ respectively. Then
	\begin{equation*} 
		\int\limits_{U(d)} \exp\{z\tr A\cConjMat{U}BU\}d\mu(U) = \bigg(\prod\limits_{j = 1}^{d - 1} j!\bigg) \frac{\det\{\exp(za_jb_k)\}_{j,k = 1}^d}{z^{(d^2 - d)/2}\Vanddet(A)\Vanddet(B)},
	\end{equation*}
	where $z$ is some constant, $\mu$ is a Haar measure, and $\Vanddet(A) = \prod\limits_{j > k}(a_j - a_k)$.
\end{proposition}
For the proof see, e.g., \cite[Appendix 5]{Me:91}.

Applying the Harish-Chan\-dra/It\-syk\-son--Zuber formula to \eqref{last asympt} we obtain
\begin{equation*} 
	\CF_m = C\mathsf{k}_n e^{-m\log d_1(\seccumul)} \frac{\det \{e^{\zeta_j\cConjScl{\zeta}_k}\}_{j,k = 1}^m}{\abs{\Vanddet(\cMatPos)}^2} (1 + o(1)),
\end{equation*}
which in combination with \eqref{F_1 behavior} yields the result of Theorem \ref{th1}.
\subsection{General case}

In the general case the proof proceeds by the same scheme as in the case of zero high cumulants. In this subsection we focus on the crucial distinctions from the partial case considered above and refine the corresponding assertions from the previous subsection. 


At the point we are ready to generalize Lemma \ref{lem:f(UL V^*) expansion}.
\begin{lemma}
	Let $\normsized{\tilde{\Lambda}} + \normsized[\big]{\hat{\tilde{\cSetMatA}}} \le \log n$. Then uniformly in $U$ and $V$
	\begin{multline} \label{f expansion gen}
	f(U(\Lambda_0 + n^{-1/2}\tilde{\Lambda})\cConjMat{V}, n^{-1/2}\hat{\tilde{\cSetMatA}}) = {- m\lambda_0^2} + n^{-1/2} \tr (\conj{z}_0\cMatPos + z_0\cConjMat{\cMatPos}) \\
	\shoveright{- \frac{1}{2n} \tr(2\lambda_0\tilde{\Lambda} + \conj{z}_0\cMatPos_U + z_0\cConjMat{\cMatPos_V})^2 + \frac{1}{n} \tr\cMatPos_U\cConjMat{\cMatPos_V}} \\
	\shoveright{- \frac{1}{2n}\tr\Bigl[(1 - \abs{\seccumul}\conj{z}_0^2)\cConjMat{\tilde{B}_{2,0}}\tilde{B}_{2,0} + (1 - \abs{\seccumul}z_0^2)\cConjMat{\tilde{B}_{0,2}}\tilde{B}_{0,2}} \\
	\shoveright{- \abs{\seccumul}\lambda_0^2\tilde{B}_{0,2}\tilde{B}_{2,0} - \abs{\seccumul}\lambda_0^2\cConjMat{\tilde{B}_{2,0}}\cConjMat{\tilde{B}_{0,2}}\Bigr] - \frac{1}{n}\norm{\tilde{\cSetMatA}_{>1}}^2} \\
	+ n^{-1}\lambda_0^2\sqrt{\cumul{2}{2}} \tr\bigl[(\wedge^2V\cConjMat{U})\tilde{\cMatA}_2 + \cConjMat{\tilde{\cMatA}_2}(\wedge^2U\cConjMat{V})\bigr] + O\big(n^{-3/2}\log^3 n\big)\lefteqn{,}
	\end{multline}
	where we keep the notations of Lemma \ref{lem:f(UL V^*) expansion} and $\wedge^2 B$ is the second exterior power of a linear operator $B$ (see \cite{Vi:03} for the definition and properties of an exterior power of a linear operator). 
\end{lemma}
\begin{proof}
	Differently from the previous subsection the function $f$ has an additional term $n^{-1/2} \tilde{h}(\cSetMatA_1, \cSetMatA_2) + n^{-1}\mathtt{p}_c(\hat{\cSetMatA})$ under the logarithm
	, where $\tilde{h}$ is defined in~\eqref{tilde h def} and $\mathtt{p}_c$ is a polynomial such that $\mathtt{p}_c(0) = 0$. 
	Therefore, the contribution of the term $n^{-1}\mathtt{p}_c(n^{-1/2}\hat{\tilde{\cSetMatA}})$ is $O\big(n^{-3/2}\log n\big)$. 
	Hence, it remains to determine the contribution of the term $n^{-1/2} \tilde{h}(\cSetMatA_1, \cSetMatA_2)$.
	
	\begin{multline}\label{tilde h before cv}
		n^{-\frac{1}{2}} \tilde{h}(n^{-\frac{1}{2}}\tilde{\cSetMatA}_1, n^{-\frac{1}{2}}\tilde{\cSetMatA}_2) = n^{-1} \tilde{h}(n^{-\frac{1}{2}}\tilde{\cSetMatA}_1, \tilde{\cSetMatA}_2) \\
		= -\frac{1}{n}\int \sum\limits_{p + s = 4} \left(\tr \tilde{\matA}_{p,s} \tilde{\cMatA}_{p,s} + \tr \cConjMat{\tilde{\cMatA}_{p,s}} \matA_{p,s}\right) e^{-\frac{1}{2} \transp{\aVecD}F\left(\frac{1}{\sqrt{n}}\cSetMatA_1\right)\aVecD} d\aConjMat{\aVecBt} d\aVecBt d\aConjMat{\aVecCt} d\aVecCt,
	\end{multline}
	where $\aVecD$ is defined in \eqref{rho def}, $F$ is defined in \eqref{A def}, $\tilde{\matA}_{p,s}$ and $\matA_{p,s}$ are defined by \eqref{chi def}. 
	
	Let us change the variables $\tilde{\aVecB} = \cConjMat{U}\aVecBt$, $\aConjMat{\tilde{\aVecB}} = \aConjMat{\aVecBt}U$, $\tilde{\aVecC} = \cConjMat{V}\aVecCt$, $\aConjMat{\tilde{\aVecC}} = \aConjMat{\aVecCt}V$. We have
	\begin{equation}\label{cv1}
	\begin{split}
	\frac{1}{\sqrt{\cumul{p}{s}}}\,\sclA_{\alpha\beta}^{(p,s)} &= \prod\limits_{q = 1}^{p} \aSclB_{\alpha_{q}} \prod\limits_{r = 1}^{s} \aConjScl{\aSclC_{\beta_{r}}}  = \prod\limits_{q = 1}^{p} (U\tilde{\aVecB})_{\alpha_{q}} \prod\limits_{r = 1}^{s} (\aConjMat{\tilde{\aVecC}}\cConjMat{V})_{\beta_{r}} \\
	&= \prod\limits_{q = 1}^{p} \sum\limits_{\gamma_q = 1}^m u_{\alpha_q\gamma_q} \tilde{\aSclB}_{\gamma_{q}} \prod\limits_{r = 1}^{s} \sum\limits_{\delta_r = 1}^m \aConjScl{\tilde{\aSclC}_{\delta_r}} \cConjScl{v}_{\beta_{r}\delta_r} \\
	&=: \sum\limits_{\gamma \in \indexset_{m, p}} \sum\limits_{\delta \in \indexset_{m, s}} \cSclB_{\alpha\beta\gamma\delta}^{(p,s)} \prod\limits_{q = 1}^{p} \tilde{\aSclB}_{\gamma_{q}} \prod\limits_{r = 1}^{s} \aConjScl{\tilde{\aSclC}_{\delta_r}},
	\end{split}
	\end{equation}
	where $\cSclB_{\alpha\beta\gamma\delta}^{(p,s)}$ just denotes the coefficient at $\prod\limits_{q = 1}^{p} \tilde{\aSclB}_{\gamma_{q}} \prod\limits_{r = 1}^{s} \aConjScl{\tilde{\aSclC}_{\delta_r}}$.
	Similarly
	\begin{equation}\label{cv2}
	\frac{1}{\sqrt{\conj{\cumul{p}{s}}}}\,\tilde{\sclA}_{\beta\alpha}^{(p,s)} = \sum\limits_{\gamma \in \indexset_{m, p}} \sum\limits_{\delta \in \indexset_{m, s}} \tilde{\cSclB}_{\beta\alpha\delta\gamma}^{(p,s)}  \prod\limits_{r = s}^{1} \tilde{\aSclC}_{\delta_r} \prod\limits_{q = p}^{1} \aConjScl{\tilde{\aSclB}_{\gamma_{q}}}
	\end{equation}
	Besides,
	\begin{equation}\label{cv3}
	\transp{\aVecD}F\aVecD = \transp{\tilde{\aVecD}}F_0\tilde{\aVecD} + O(n^{-1/2}\log n),
	\end{equation}
	where $F_0$ is defined in \eqref{F_0,F_1 def} and
	\begin{align*}
		\tilde{\aVecD} = 
		\begin{pmatrix}
			\transp{\bigl(\aConjMat{\tilde{\aVecB}}\bigr)} \\
			\transp{\bigl(\aConjMat{\tilde{\aVecC}}\bigr)} \\
			\tilde{\aVecB} \\
			\tilde{\aVecC}
		\end{pmatrix}.
	\end{align*}
	The ``measure'' changes as follows 
	\begin{equation}\label{cv_diff}
	\begin{split}
	d\aConjMat{\aVecBt} d\aVecBt d\aConjMat{\aVecCt} d\aVecCt &= \det\nolimits^{-1} U \det\nolimits^{-1} \cConjMat{U} \det\nolimits^{-1} V \det\nolimits^{-1} \cConjMat{V} d\aConjMat{\tilde{\aVecB}} d\tilde{\aVecB} d\aConjMat{\tilde{\aVecC}} d\tilde{\aVecC} \\
	&= d\aConjMat{\tilde{\aVecB}} d\tilde{\aVecB} d\aConjMat{\tilde{\aVecC}} d\tilde{\aVecC}.
	\end{split}
	\end{equation}
	Eventually, substitution of \eqref{cv1}--\eqref{cv_diff} into \eqref{tilde h before cv} yields
	\begin{equation} \label{tilde h after cv}
	\begin{split}
	n^{-1} \tilde{h}(n^{-\frac{1}{2}}\tilde{\cSetMatA}_1, \tilde{\cSetMatA}_2) &= -\frac{1}{n} \int e^{-\frac{1}{2}\aConjMat{\tilde{\aVecD}}F_0\tilde{\aVecD}} d\aConjMat{\tilde{\aVecB}} d\tilde{\aVecB} d\aConjMat{\tilde{\aVecC}} d\tilde{\aVecC} \\
	&\qquad{}\times \sum\limits_{p + s = 4} \sum\limits_{\substack{\alpha, \gamma \in \indexset_{m,p} \\ \beta, \delta \in \indexset_{m,s}}} \bigg( \sqrt{\conj{\cumul{p}{s}}} \tilde{\cSclB}_{\beta\alpha\delta\gamma}^{(p,s)} \tilde{\cSclA}_{\alpha\beta}^{(p,s)} \prod\limits_{r = s}^{1} \tilde{\aSclC}_{\delta_r} \prod\limits_{q = p}^{1} \aConjScl{\tilde{\aSclB}_{\gamma_{q}}} \\
	&\qquad{}+ \sqrt{\cumul{p}{s}} \cConjScl{\tilde{\cSclA}}_{\alpha\beta}^{(p,s)} \cSclB_{\alpha\beta\gamma\delta}^{(p,s)} \prod\limits_{q = 1}^{p} \tilde{\aSclB}_{\gamma_{q}} \prod\limits_{r = 1}^{s} \aConjScl{\tilde{\aSclC}_{\delta_r}} \bigg) \\
	&\qquad{}+ O\big(n^{-3/2}\log^3 n\big)
	\end{split}
	\end{equation}
	uniformly in $U$ and $V$.
	
	The integration in \eqref{tilde h after cv} can be performed over $\tilde{\aSclB}_{j}$, $\tilde{\aSclC}_{j}$ separately for every~$j$ due to the structure of $F_0$. Thus it remains to compute the integrals of the form
	\begin{equation*}
		\int \prod\limits_{q = 1}^{p} \tilde{\aSclB}_{\gamma_{q}} \prod\limits_{r = 1}^{s} \aConjScl{\tilde{\aSclC}_{\delta_r}} \exp\left\{\conj{z}_0\tilde{\aSclC}_j\aConjScl{\tilde{\aSclC}_j} + \lambda_0 \tilde{\aSclC}_j\aConjScl{\tilde{\aSclB}_j} +z_0\tilde{\aSclB}_j\aConjScl{\tilde{\aSclB}_j} - \lambda_0 \tilde{\aSclB}_j\aConjScl{\tilde{\aSclC}_j}\right\} d\aConjScl{\tilde{\aSclB}_{j}} d\tilde{\aSclB}_{j} d\aConjScl{\tilde{\aSclC}_{j}} d\tilde{\aSclC}_{j}
	\end{equation*}
	Furthermore, expanding the exponent into series, one can observe that all the integrals are non-zero only if $p = s = 2$ and $\gamma = \delta$. Moreover,
	\begin{equation*}
		\begin{split}
			\int \tilde{\aSclB}_{j}\aConjScl{\tilde{\aSclC}_{j}} e^{{z_0\tilde{\aSclC}_j\aConjScl{\tilde{\aSclC}_j} + \lambda_0 \tilde{\aSclC}_j\aConjScl{\tilde{\aSclB}_j} +z_0\tilde{\aSclB}_j\aConjScl{\tilde{\aSclB}_j} - \lambda_0 \tilde{\aSclB}_j\aConjScl{\tilde{\aSclC}_j}}} d\aConjScl{\tilde{\aSclB}_{j}} d\tilde{\aSclB}_{j} d\aConjScl{\tilde{\aSclC}_{j}} d\tilde{\aSclC}_{j} &= -\lambda_0,
			\\
			\int \tilde{\aSclC}_{j}\aConjScl{\tilde{\aSclB}_{j}} e^{{z_0\tilde{\aSclC}_j\aConjScl{\tilde{\aSclC}_j} + \lambda_0 \tilde{\aSclC}_j\aConjScl{\tilde{\aSclB}_j} +z_0\tilde{\aSclB}_j\aConjScl{\tilde{\aSclB}_j} - \lambda_0 \tilde{\aSclB}_j\aConjScl{\tilde{\aSclC}_j}}} d\aConjScl{\tilde{\aSclB}_{j}} d\tilde{\aSclB}_{j} d\aConjScl{\tilde{\aSclC}_{j}} d\tilde{\aSclC}_{j} &= \lambda_0,
			\\
			\int e^{{z_0\tilde{\aSclC}_j\aConjScl{\tilde{\aSclC}_j} + \lambda_0 \tilde{\aSclC}_j\aConjScl{\tilde{\aSclB}_j} +z_0\tilde{\aSclB}_j\aConjScl{\tilde{\aSclB}_j} - \lambda_0 \tilde{\aSclB}_j\aConjScl{\tilde{\aSclC}_j}}} d\aConjScl{\tilde{\aSclB}_{j}} d\tilde{\aSclB}_{j} d\aConjScl{\tilde{\aSclC}_{j}} d\tilde{\aSclC}_{j} &= 1.
		\end{split}
	\end{equation*}
	The last thing we need is values of $\cSclB_{\alpha\beta\gamma\delta}^{(2,2)}$ and $\tilde{\cSclB}_{\beta\alpha\delta\gamma}^{(2,2)}$. The formula~\eqref{cv1} implies
	\begin{equation*}
	\cSclB_{\alpha\beta\gamma\delta}^{(2,2)} = (\wedge^2 U)_{\alpha\gamma} (\wedge^2 \cConjMat{V})_{\delta\beta}.
	\end{equation*}
	Similarly
	\begin{equation*}
	\tilde{\cSclB}_{\beta\alpha\delta\gamma}^{(2,2)} = (\wedge^2 V)_{\beta\gamma} (\wedge^2 \cConjMat{U})_{\delta\alpha}.
	\end{equation*}
	Finally
	\begin{equation*}
		\begin{split}
			\frac{1}{n} 
			\tilde{h}(n^{-\frac{1}{2}}\tilde{\cSetMatA}_1, \tilde{\cSetMatA}_2) &= 
			\frac{1}{n} 
			\lambda_0^2\sqrt{\cumul{2}{2}} (\tr(\wedge^2\cConjMat{U})\tilde{Q}_2(\wedge^2V) + \tr(\wedge^2\cConjMat{V})\cConjMat{\tilde{Q}_2}(\wedge^2U)) + o\left(\frac{1}{n} 
			\right) \\
			&= n^{-1}\lambda_0^2\sqrt{\cumul{2}{2}} (\tr(\wedge^2V\cConjMat{U})\tilde{Q}_2 + \tr\cConjMat{\tilde{Q}_2}(\wedge^2U\cConjMat{V})) \\
			&\qquad{}+ O\big(n^{-3/2}\log^3 n\big).
		\end{split}
	\end{equation*}
	The above relation completes the proof of \eqref{f expansion gen}.
\end{proof}
Lemma~\ref{lem:est for Re f} is still valid in the general case, despite the proof needs some insignificant changes due to a non-zero term $n^{-1/2} \tilde{h}(\cSetMatA_1, \cSetMatA_2)$.

Following the proof in the Gaussian case one can see that \eqref{n-indep int} transforms into
\begin{equation*}
\begin{split}
\CF_m &= C\mathsf{k}_n \int\limits_{\sqrt{n}\stpointsnbh} \Vanddet^2(\tilde{\Lambda}) g(\cSetMatA_*) e^{-c_m f(\Lambda_0, 0)} d\mu(U) d\mu(V) d\tilde{\Lambda} d\hat{\tilde{\cSetMatA}} (1 + o(1)) \\
&\quad \times \exp\Bigl\lbrace -\frac{1}{2}\tr(2\lambda_0\tilde{\Lambda} + \bar{z}_0\cMatPos_U + z_0\cMatPos_\cConjMat{V})^2 + \tr\cMatPos_U\cMatPos_\cConjMat{V} - \norm{\tilde{\cSetMatA}_{>1}}^2 \\
&\quad\phantom{\times \exp\Bigl\lbrace} - \frac{1}{2}\tr\Bigl[(1 - \abs{\seccumul}\conj{z}_0^2)\cConjMat{\tilde{B}_{2,0}}\tilde{B}_{2,0} + (1 - \abs{\seccumul}z_0^2)\cConjMat{\tilde{B}_{0,2}}\tilde{B}_{0,2} \\
&\quad\phantom{\times \exp\Bigl\lbrace - \frac{1}{2}\tr\Bigl[} - \abs{\seccumul}\lambda_0^2\tilde{B}_{0,2}\tilde{B}_{2,0} - \abs{\seccumul}\lambda_0^2\cConjMat{\tilde{B}_{2,0}}\cConjMat{\tilde{B}_{0,2}}\Bigr] \\
&\quad\phantom{\times \exp\Bigl\lbrace} + \lambda_0^2\sqrt{\cumul{2}{2}} \tr\bigl[(\wedge^2V\cConjMat{U})\tilde{\cMatA}_2 + \cConjMat{\tilde{\cMatA}_2}(\wedge^2U\cConjMat{V})\bigr]\Bigr\rbrace,
\end{split}
\end{equation*}
where $\mathsf{k}_n$ is defined in \eqref{K_n def}. The Gaussian integration over $\hat{\tilde{\cSetMatA}}$ yields
\begin{equation*}
\begin{split}
\CF_m &= C\mathsf{k}_n d_1(\seccumul)^{-m} \int \Vanddet^2(\tilde{\Lambda}) d\mu(U) d\mu(V) d\tilde{\Lambda} (1 + o(1)) \\
&\qquad\qquad\quad \times \exp\Bigl\lbrace -\frac{1}{2}\tr(2\lambda_0\tilde{\Lambda} + \bar{z}_0\cMatPos_U + z_0\cMatPos_\cConjMat{V})^2 + \tr\cMatPos_U\cMatPos_\cConjMat{V} \\
&\qquad\qquad\quad\phantom{\times \exp\Bigl\lbrace} + \lambda_0^4\cumul{2}{2} \tr\bigl[(\wedge^2V\cConjMat{U})(\wedge^2U\cConjMat{V})\bigr] \Bigr\rbrace.
\end{split}
\end{equation*}
Note that $\wedge^2V\cConjMat{U}$ and $\wedge^2U\cConjMat{V}$ are mutually inverse matrices. 
Therefore,
\begin{equation*}
\begin{split}
\CF_m &= C\mathsf{k}_n d_1(\seccumul)^{-m}\exp\left\{\frac{m^2 - m}{2}\lambda_0^4\cumul{2}{2} \right\} \int \Vanddet^2(\tilde{\Lambda}) d\mu(U) d\mu(V) d\tilde{\Lambda} (1 + o(1)) \\
&\qquad\qquad\quad \times \exp\Bigl\lbrace -\frac{1}{2}\tr(2\lambda_0\tilde{\Lambda} + \bar{z}_0\cMatPos_U + z_0\cMatPos_\cConjMat{V})^2 + \tr\cMatPos_U\cMatPos_\cConjMat{V} \Bigr\rbrace.
\end{split}
\end{equation*}
The last formula differs from~\eqref{n-indep int 1} only by a factor $\exp\left\{\frac{m^2 - m}{2}\lambda_0^4\realcumul_4\right\}$. Hence, 
there are no differences in further proof up to 
this factor.



\section*{Acknowledgments.} The author is grateful to Prof.\ M.\ Shcherbina for the statement of the problem and fruitful discussions.


\subsection*{Supports.} The author is supported in part by the Akhiezer Foundation scholarship and by the NASU scholarship for young scientists.







\begin{thebibliography}{99}

\bibitem{Af:16}
I.~Afanasiev,
\newblock {\emph{On the Correlation Functions of the Characteristic Polynomials of
the Sparse Hermitian Random Matrices,}}
\newblock {J. Stat. Phys.} \textbf{163(2)} (2016), 324--356.

\bibitem{Af:19}
I.~Afanasiev,
\newblock {\emph{On the Correlation Functions of the Characteristic Polynomials of
  Non-Hermitian Random Matrices with Independent Entries,}}
\newblock {J. Stat. Phys.} \textbf{176(6)} (2019), 1561--1582.

\bibitem{Af:20}
I.~Afanasiev,
\newblock {\emph{On the Correlation Functions of the Characteristic Polynomials of
  Real Random Matrices with Independent Entries,}}
\newblock {Zh. Mat. Fiz. Anal. Geom.} \textbf{16(2)} (2020), 91--118.

\bibitem{Ak-Ka:07}
G.~Akemann and E.~Kanzieper,
\newblock {\emph{Integrable structure of Ginibre's ensemble of real random matrices
  and a Pfaffian integration theorem,}}
\newblock {J. Stat. Phys.} \textbf{129(5-6)} (2007), 1159--1231.

\bibitem{Ak-Ve:03}
G.~Akemann and G.~Vernizzi, \emph{Characteristic polynomials of complex random matrix
  models,}
\newblock Nucl. Phys. B \textbf{660(3)} (2003), 532--556

\bibitem{Ba-Er:17}
Z. Bao and L. Erd{\H o}s, \emph{Delocalization for a class of random block band
  matrices,}
\newblock Probab. Theory Relat. Fields \textbf{167} (2017), 673--776

\bibitem{Be:87}
F.~A. Berezin,
\newblock {\emph{Introduction to superanalysis}},
\newblock Number~9 in Math. Phys. Appl. Math. D.~Reidel Publishing Co.,
  Dordrecht, 1987.
\newblock Edited and with a foreword by A. A. Kirillov. With an appendix by V.
  I. Ogievetsky. Translated from the Russian by J. Niederle and R. Koteck\'y.
  Translation edited by Dimitri Le\u{\i}tes.

\bibitem{Bo-Ch:12}
C.~Bordenave and D.~Chafa\"\i,
\newblock \emph{Around the circular law,}
\newblock {Probab. Surv.} \textbf{9} (2012), 1--89.

\bibitem{Bo-Si:09}
A.~Borodin and C.~D. Sinclair,
\newblock {\emph{The Ginibre Ensemble of Real Random Matrices and its Scaling
  Limits,}}
\newblock {Comm. Math. Phys.} \textbf{291} (2009), 177--224.

\bibitem{Bo-St:06}
A.~Borodin and E.~Strahov,
\newblock \emph{Averages of characteristic polynomials in random matrix theory,}
\newblock {Comm. Pure Appl. Math.} \textbf{59(2)} (2006), 161--253.

\bibitem{Br-Hi:00}
E.~Br\'ezin and S.~Hikami,
\newblock \emph{Characteristic polynomials of random matrices,}
\newblock {Comm. Math. Phys.} \textbf{214} (2000), 111--135.

\bibitem{Br-Hi:01}
E.~Br\'ezin and S.~Hikami.
\newblock \emph{Characteristic polynomials of real symmetric random matrices,}
\newblock {Comm. Math. Phys.} \textbf{223} (2001), 363--382.

\bibitem{Ci-Er-Sc:20}
G.~Cipolloni, L.~Erd{\H o}s and D.~Schr\"oder,
\newblock \emph{Optimal lower bound on the least singular value of the shifted
  Ginibre ensemble,}
\newblock {Prob. Math. Physics} \textbf{1} (2020), 101--146, \arXiv{1908.01653v3}.

\bibitem{Ci-Er-Sc:21_cCLT}
G. Cipolloni, L. Erd{\H o}s and D. Schr\"oder, \emph{Central limit theorem for
  linear eigenvalue statistics of non-Hermitian random matrices,}
\newblock Probab. Theory Related Fields \textbf{179} (2021), 1--28, \arXiv{1912.04100}.

\bibitem{Ci-Er-Sc:21_rCLT}
G.~Cipolloni, L.~Erd{\H o}s and D.~Schr\"oder,
\newblock \emph{Fluctuation around the circular law for random matrices with real entries,}
\newblock {Electron. J. Prob.}, \textbf{24} (2021), Paper No. 24, \arXiv{2002.02438v1}.

\bibitem{Ci-Er-Sc:22}
G.~Cipolloni, L.~Erd{\H o}s and D.~Schr\"oder,
\newblock \emph{Edge universality for non-Hermitian random matrices,}
\newblock to appear in Comm. Pure Appl. Math. (2022), DOI \href{https://doi.org/10.1002/cpa.22028}{10.1002/cpa.22028}, \arXiv{1908.00969v2}. 

\bibitem{Di-La:17}
M. Disertori and M. Lager, \emph{Density of States for Random Band Matrices in Two
  Dimensions},
\newblock Ann. Henri Poincar\'e \textbf{18} (2017), 2367--2413

\bibitem{Di-La:20}
M. Disertori and M. Lager, \emph{Supersymmetric Polar Coordinates with applications
  to the Lloyd model},
\newblock Math. Phys. Anal. Geom. \textbf{23(1)} (2020), Paper No. 2, 21 pp.

\bibitem{Di-Lo-So:21}
M.~Disertori, M.~Lohmann and S.~Sodin,
\newblock \emph{The density of states of 1D random band matrices via a
  supersymmetric transfer operator,}
\newblock J. Spectr. Theory \textbf{11(1)} (2021), 125--191

\bibitem{Di-Me-Ro:14}
M. Disertori, F. Merkl and S. Rolles, \emph{Localization for a Nonlinear Sigma Model
  in a Strip Related to Vertex Reinforced Jump Processes},
\newblock Commun. Math. Phys. \textbf{332} (2014), 783--825

\bibitem{Di-Sp-Zi:10}
M.~Disertori, T.~Spencer and M.~R. Zirnbauer,
\newblock \emph{Quasi-diffusion in a 3D supersymmetric hyperbolic sigma model,}
\newblock {Comm. Math. Phys.} \textbf{300(2)} (2010), 435--486.

\bibitem{Ed:97}
A. Edelman, \emph{The probability that a random real Gaussian matrix has k real eigenvalues, related distributions, and the circular law}.
\newblock J. Multivariate Anal. \textbf{60(2)} (1997), 203--232

\bibitem{Ef:98}
K.~Efetov,
\newblock \emph{Supersymmetry in disorder and chaos},
\newblock Cambridge University Press, Cambridge, 1997.

\bibitem{Ef:83}
K.~B. Efetov,
\newblock \emph{Supersymmetry and theory of disordered metals,}
\newblock {Adv. in Physics} \textbf{32(1)} (1983), 53--127.

\bibitem{Fo-Na:07}
P. Forrester and T. Nagao, \emph{Eigenvalue statistics of the real Ginibre ensemble,}
\newblock Phys. Rev. Lett. \textbf{99} (2007), 050603

\bibitem{Fo:99}
P.~J.~Forrester, \emph{Fluctuation formula for complex random matrices,}
\newblock J. Phys. A \textbf{32(13)} (1999), L159--L163

\bibitem{Fy:02}
Y.~V.~Fyodorov, \emph{Negative moments of characteristic polynomials of random
  matrices: Ingham–Siegel integral as an alternative to
  Hubbard–Stratonovich transformation},
\newblock Nucl. Phys. B \textbf{621(3)} (2002), 643--674

\bibitem{Fy-Kh:99}
Y.~V.~Fyodorov and B.~A.~Khoruzhenko, \emph{Systematic Analytical Approach to
  Correlation Functions of Resonances in Quantum Chaotic Scattering},
\newblock Phys. Rev. Lett. \textbf{83(1)} (1999), 65--68

\bibitem{Fy-Mi:91}
Y.~V. Fyodorov and A.~D. Mirlin,
\newblock \emph{Localization in ensemble of sparse random matrices,}
\newblock {Phys. Rev. Lett.} \textbf{67} (1991), 2049--2052.

\bibitem{Fy-So:03}
Y.~V.~Fyodorov and  H.-J.~Sommers
\newblock \emph{Random matrices close to Hermitian or unitary: overview of methods
  and results},
\newblock {J. Phys. A} \textbf{36(12)} (2003), 3303--3347.

\bibitem{Fy-St:03}
Y.~V. Fyodorov and E.~Strahov,
\newblock \emph{An exact formula for general spectral correlation function of random
  Hermitian matrices. Random matrix theory},
\newblock {J. Phys. A} \textbf{36(12)} (2003), 3203--3214.

\bibitem{Gin:65}
J.~Ginibre,
\newblock \emph{Statistical ensembles of complex, quaternion, and real matrices,}
\newblock {J. Math. Phys.} \textbf{6} (1965), 440--449.

\bibitem{Gir:84}
V.~L. Girko,
\newblock \emph{The circular law,}
\newblock {Teor. Veroyatn. Primen.} \textbf{29(4)} (1984), 669--679.

\bibitem{Gir:94}
V.~L. Girko,
\newblock \emph{The circular law: ten years later,}
\newblock {Random Oper. Stoch. Equ.} \textbf{2(3)} (1994), 235--276.

\bibitem{Gir:04_1}
V.~L. Girko,
\newblock \emph{The strong circular law. Twenty years later. I,}
\newblock {Random Oper. Stoch. Equ.} \textbf{12(1)} (2004), 49--104.

\bibitem{Gir:04_2}
V.~L. Girko,
\newblock \emph{The strong circular law. Twenty years later. II,}
\newblock {Random Oper. Stoch. Equ.} \textbf{12(3)} (2004), 255--312.

\bibitem{Gir:05}
V.~L. Girko,
\newblock \emph{The circular law. Twenty years later. III,}
\newblock {Random Oper. Stoch. Equ.} \textbf{13(1)} (2005), 53--109.

\bibitem{Gu:15}
T.~Guhr,
\newblock \emph{Supersymmetry,}
\newblock {{The
  Oxford Handbook of Random Matrix Theory}} (eds. G.~Akemann, J.~Baik and P.~D. Francesco), Oxford
  university press, 2015, chapter~7, 135--154.

\bibitem{Hu:63}
L.~K.~Hua,
\newblock \emph{{Harmonic Analysis of Functions of Several Complex Variables in the Classical Domains}},
\newblock American Mathematical Society, Providence, RI, 1963.

\bibitem{Ko:15}
P.~Kopel,
\newblock \emph{Linear Statistics of Non-Hermitian Matrices Matching the Real or
  Complex Ginibre Ensemble to Four Moments,}
\newblock preprint, \arXiv{1510.02987v1} 

\bibitem{superbos}
P.~Littelmann, H.-J. Sommers and M.R. Zirnbauer,
\newblock \emph{Superbosonization of invariant random matrix ensembles,}
\newblock {Comm. Math. Phys.}, \textbf{283} (2008), 343--395.

\bibitem{Me:67}
M.~L. Mehta,
\newblock \emph{Random matrices and the statistical theory of energy levels},
\newblock Academic Press, New York--London, 1967.

\bibitem{Me:91}
M.~L. Mehta,
\newblock \emph{Random Matrices}, second edn.
\newblock Academic Press Inc., Boston, 1991.

\bibitem{Mi-Fy:91}
A.~D. Mirlin and Y.~V. Fyodorov,
\newblock \emph{Universality of level correlation function of sparse random matrices,}
\newblock {J. Phys. A} \textbf{24} (1991), 2273--2286.

\bibitem{OR-Re:16}
S.~O'Rourke and D.~Renfrew,
\newblock \emph{Central limit theorem for linear eigenvalue statistics of elliptic
  random matrices,}
\newblock {J. Theoret. Probab.} \textbf{29} (2016), 1121--1191.

\bibitem{Re-Ki-Gu-Zi:12}
C.~Recher, M.~Kieburg, T.~Guhr and M.~R. Zirnbauer,
\newblock \emph{Supersymmetry approach to Wishart correlation matrices: Exact
  results,}
\newblock {J. Stat. Phys.} \textbf{148(6)} (2012), 981--998.

\bibitem{Ri-Si:06}
B. Rider and J. Silverstein, \emph{Gaussian fluctuations for non-Hermitian random
  matrix ensembles}.
\newblock Ann. Probab. \textbf{34(6)} (2006), 2118--2143

\bibitem{Ri-Vi:07}
B. Rider and B. Virag, \emph{The noise in the circular law and the Gaussian free
  field}.
\newblock Int. Math. Res. Not. IMRN \textbf{2} (2007), Art. ID rnm006, 33 pp

\bibitem{Sha:13}
M. Shamis, \emph{Density of states for Gaussian unitary ensemble, Gaussian
  orthogonal ensemble, and interpolating ensembles through supersymmetric
  approach},
\newblock J. Math. Phys. \textbf{54} (2013), 113505

\bibitem{ShM-ShT:16}
M.~Shcherbina and T.~Shcherbina,
\newblock \emph{Transfer matrix approach to 1d random band matrices: density of
  states,}
\newblock {J. Stat. Phys.} \textbf{164(6)} (2016), 1233--1260.

\bibitem{ShM-ShT:17}
M.~Shcherbina and T.~Shcherbina,
\newblock \emph{Characteristic polynomials for 1D random band matrices from the
  localization side,}
\newblock {Comm. Math. Phys.} \textbf{351(3)} (2017), 1009--1044.

\bibitem{ShM-ShT:18}
M.~Shcherbina and T.~Shcherbina,
\newblock \emph{Universality for 1d random band matrices: sigma-model
  approximation,}
\newblock {J. Stat. Phys.} \textbf{172(2)} (2018), 627--664.

\bibitem{ShT:11}
T.~Shcherbina,
\newblock \emph{On the correlation function of the characteristic polynomials of the
  Hermitian Wigner ensemble,}
\newblock {Comm. Math. Phys.} \textbf{308} (2011), 1--21.

\bibitem{ShT:13}
T.~Shcherbina,
\newblock \emph{On the correlation functions of the characteristic polynomials of
  the Hermitian sample covariance matrices,}
\newblock {Probab. Theory Related Fields} \textbf{156} (2013), 449--482.

\bibitem{St-Fy:03}
E.~Strahov and Y.~V. Fyodorov,
\newblock \emph{Universal results for correlations of characteristic polynomials:
  Riemann-Hilbert approach,}
\newblock {Comm. Math. Phys.} \textbf{241(2-3)} (2003), 343--382.

\bibitem{Ta-Vu:10}
T.~Tao and V.~Vu,
\newblock \emph{Random matrices: universality of ESDs and the circular law,}
\newblock {Ann. Probab.} \textbf{38(5)} (2010), 2023--2065.
\newblock With an appendix by Manjunath Krishnapur.

\bibitem{Ta-Vu:15}
T.~Tao and V.~Vu,
\newblock \emph{Random matrices: universality of local spectral statistics of
  non-Hermitian matrices,}
\newblock {Ann. Probab.} \textbf{43(2)} (2015), 782--874.

\bibitem{Vi:03}
E.~B. Vinberg,
\newblock \emph{{A Course in Algebra}},
\newblock American Mathematical Society, Providence, RI, 2003.

\bibitem{We-Wo:19}
C.~Webb and M.D.~Wong, \emph{On the moments of the characteristic polynomial of a
  Ginibre random matrix},
\newblock Proc. Lond. Math. Soc. (3) \textbf{118(5)} (2019), 1017--1056

\bibitem{Zi:06}
M.~R.~Zirnbauer, \emph{The supersymmetry method of random matrix theory.}
\newblock In: Encyclopedia of mathematical physics, vol.~5, pp. 151--160.
  Elsevier (2006).
\newblock \arXiv{math-ph/0404057}



\end{thebibliography}
\end{document}